\documentclass[runningheads]{llncs}
\usepackage{amsmath} 
\usepackage{graphicx}
\usepackage[T1]{fontenc}
\usepackage{pxfonts}
\def\calA{\mathcal{A}}

\newtheorem{observation}{Observation}
\makeatletter
\renewcommand\subsubsection{\@startsection{subsubsection}{3}{\z@}%
                       {-18\p@ \@plus -4\p@ \@minus -4\p@}%
                       {0.5em \@plus 0.22em \@minus 0.1em}%
                       {\normalfont\normalsize\bfseries\boldmath}}
\makeatother
\usepackage{fullpage}
\usepackage{doi}

\begin{document}

\title{The Connected $k$-Vertex One-Center Problem on Graphs\thanks{A preliminary version of this paper will appear in \textit{Proceedings of the 19th International Conference and Workshops on Algorithms and Computation (WALCOM 2025)}. This research was supported in part by U.S. National Science Foundation under Grant CCF-2339371.}}
\author{Jingru Zhang
}
\authorrunning{J. Zhang}

\institute{Cleveland State University, Cleveland, Ohio 44115, USA\\
\email{j.zhang40@csuohio.edu}}

\maketitle             
\begin{abstract}
We consider a generalized version of the (weighted) one-center problem on graphs. Given an undirected graph $G$ of $n$ vertices and $m$ edges and a positive integer $k\leq n$, the problem aims to find a point in $G$ so that the maximum (weighted) distance from it to $k$ connected vertices in its shortest path tree(s) is minimized. No previous work has been proposed for this problem except for the case $k=n$, that is, the classical graph one-center problem. In this paper, an $O(mn\log n\log mn + m^2\log n\log mn)$-time algorithm is proposed for the weighted case, and an $O(mn\log n)$-time algorithm is presented for the unweighted case, provided that the distance matrix for $G$ is given. When $G$ is a tree graph, we propose an algorithm that solves the weighted case in $O(n\log^2 n\log k)$ time with no given distance matrix, and improve it to $O(n\log^2 n)$ for the unweighted case.
\keywords{Algorithms\and Data Structures\and Facility Locations\and $k$-Vertex One-Center\and Graphs\and Connectivity}
\end{abstract}

\section{Introduction}\label{sec:intro}
The one-center problem is a classical problem in facility locations which aims to compute the best location of a single facility on a graph network to serve customers such that the maximum (weighted) distance between the facility and all customers is minimized~\cite{ref:GoldmanMi72,ref:HandlerMi73,ref:KarivAnC79,ref:LanAl99,ref:FoulA106,ref:BenMosheEf07,ref:WangCo17,ref:HuCo23}. Due to the resource limits, it is quite natural to consider the partial version where the facility serves only $k$ neighboring customers with the minimized maximum 
transportation cost, that is, this connected $k$-vertex one-center problem. 

Let $G=(V,E)$ be an undirected graph of $n$ vertices and $m$ edges where each vertex $v\in V$ has a weight $w_v>0$ and each edge $e\in E$ is of length $l(e)> 0$. For any two vertices $u,v\in V$, let $e(u,v)$ be the edge between them. By considering $e(u,v)$ as a line segment of length $l(e(u,v))$, we can talk about “points” on it. Formally, a point $p=(u,v,t(p))$ on edge $e(u,v)$ is characterized by being located at a distance of $t(p)\leq l(e(u,v))$ from the vertex $u$. We say that $p$ is interior of $e(u,v)$ if $0<t(p)<l(e(u,v))$. For any two points $x,y$ of $G$, the distance $d(x,y)$ between them is defined as the length of their shortest path(s) $\pi(x,y)$ in $G$.

A point $x$ of $G$ may have multiple shortest path trees, and denote their set by $G(x)$. Let $T$ represent a tree graph and $T^k$ denote a tree of size $k$. A subgraph of $G$ is called a \textit{$k$-subtree} if and only if it is a tree of size $k$. Let $G^k(x)$ be the set of all distinct $k$-subtrees of trees in $G(x)$. Denote by $V(G')$ the subset of all vertices in a subgraph $G'$ of $G$.

Define $\phi(x,G)$ as $\min_{T^k\in G^k(x)}\max_{v\in V(T^k)} w_vd(v,x)$. The problem aims to compute a point $x^*$ on $G$, called the \textit{partial center}, so as to minimize $\phi(x,G)$. Note that $x^*$ might be interior of an edge in $G$. 


If $G$ is a tree graph, every point of $G$ has a unique shortest path tree, i.e., $G$ itself. Clearly, in this situation, the problem is equivalent to the problem of finding a $k$-subtree of minimum (weighted) radius where $x^*$ is the (weighted) center of an optimal $k$-subtree. 

When $k=n$, $x^*$ is exactly the (weighted) center of $G$ with respect to $V$. Provided that the distance matrix is given, the center can be found in $O(mn\log n)$ time and the unweighted case can be addressed in $O(mn+n^2\log n)$ time~\cite{ref:KarivAnC79}. Additionally, for $G$ being a tree, the (weighted) one-center problem can be solved in $O(n)$ time~\cite{ref:MegiddoLi83}. 

As far as we are aware, however, this connected $k$-vertex one-center problem has not received any attention even when $G$ is a tree. In this paper, we solve this problem in $O(mn\log n\log mn+ m^2\log n\log mn)$ time and address the unweighted case more efficiently in $O(mn\log n)$ time with the given distance matrix. For $G$ being a tree, an $O(n\log^2n\log k)$-time algorithm is proposed for the weighted case and an $O(n\log^2n)$-time approach is presented for the unweighted case. 

\subsection{Related Work}
As introduced above, when $k = n$, Kariv and Hamiki~\cite{ref:KarivAnC79} proposed an $O(mn\log n)$-time algorithm for the weighted case and an $O(mn+ n^2\log n)$-time algorithm for the unweighted case, provided that the distance matrix is given. Megiddo~\cite{ref:MegiddoLi83} solved the (weighted) one-center problem on trees in linear time by the prune-and-search techniques. 

Although this partial version of the one-center problem has not been studied before, some other partial variants of the general $p$-center problem have been explored in literature. Megiddo et al.~\cite{ref:MegiddoTh83} gave an $O(n^2p)$-time algorithm to solve the maximum coverage problem which aims to place $p$ facilities on tree networks to cover maximum customers within their covering range. Berman et al.~\cite{ref:BermanTh94} considered another variant that places $p$ facilities to minimize the maximum distance between them and customers within their covering range. See other partial versions of the $p$-center problem and  variances~\cite{ref:BermanTh02,ref:BermanTh09,ref:DaskinTw10}.
 
Another most related problem is the graph maximum $t$-club problem where the goal is to find the maximum-cardinality subgraph of diameter no more than value $t$. Bourjolly et al.~\cite{ref:BourjollyAn02} revealed the NP-Hardness of this problem. Asahiro et al.~\cite{ref:AsahiroOp18} proposed an approximation algorithm of $O(n^\frac{1}{2})$ ratio, which was proved to be optimal for any $t>2$. Additionally, a constant ratio was achieved for this problem on unit disk graphs~\cite{ref:Abu-AffashAp21}. 

\subsection{Our Approach}
Denote by $\lambda^*$ the minimized objective value $\phi(x^*,G)$. We show that $\lambda^*$ belongs to a finite set that includes the following values w.r.t. every edge $e$: the one-center objective values of every two distinct vertices with the \textit{local} constraint where their centers must be on $e$, and the (weighted) distance of every vertex to its \textit{semicircular} point on $e$ which is the point on $e$ such that the vertex has two shortest paths to it without any common intermediate vertex. For the unweighted version or the version where $G$ is a tree, however, we observe that $\lambda^*$ is only relevant to those constrained one-center objective values w.r.t. each edge. 
 
For the weighted version, by forming this finite set as a set of $y$-coordinates of intersections between $O(mn)$ lines, we can adapt the line arrangement search technique~\cite{ref:ChenAn13} to find $\lambda^*$ with the assistance of our feasibility test that determines for any given value $\lambda$, whether $\lambda\geq\lambda^*$. Obviously, $\lambda\geq\lambda^*$ if there exists a point on $G$ such that it covers a $k$-subtree in its shortest path tree(s) under $\lambda$ (that is, the weighted distance from it to each vertex of a $k$-subtree in its shortest path tree(s) is no more than $\lambda$). Otherwise, $\lambda<\lambda^*$, so $\lambda$ is not feasible. 

Our feasibility test is motivated by a critical observation: $\lambda$ is feasible if and only if there exists a point in $G$ such that the largest \textit{self-inclusive} subtree covered by it in its shortest path tree(s) is of size at least $k$. Hence, our algorithm examines every edge $e$ of $G$ to decide the existence of such a point by algorithmically constructing a function that computes for every point on $e$ the size of the largest self-inclusive subtree covered by the point, which is one of our main contributions. By determining the breakpoints of these functions, the feasibility of $\lambda$ can be known in $O(mn\log n + m^2\log n)$ time. 

For the unweighted version, our key observation, that $\lambda^*$ is decided by the smallest one among the $k$-th shortest path lengths of all vertices, implies that finding the \textit{local} partial center on every edge is equivalent to a geometry problem that computes the lowest vertex on the \textit{$k$-th level} of $O(n)$ $x$-monotone polygonal chains of complexity $O(1)$. We develop an $O(n\log n)$-time algorithm for this geometry problem, which is our another main contribution, and the unweighted version is thus addressed in $O(mn\log n)$ time. 

When $G$ is a tree, we develop several data structures so that for any given point of $G$, the counting query on the largest self-inclusive subtree covered by the point can be answered in $O(\log n\log k)$ time. In addition, instead of examining every edge, we observe that only $O(n)$ points on the tree need to be examined in order to decide the feasibility of $\lambda$. These lead an $O(n\log n\log k)$-time feasibility test for the tree version. Then, by implicitly forming the set of the one-center objective values for all pairs of vertices, $\lambda^*$ can be computed in $O(n\log^2n\log k)$ time. For the unweighted case, $O(n\log n)$ sorted sets of intervertex distances are implicitly formed by employing the tree decomposition technique~\cite{ref:MegiddoAn81}, so $\lambda^*$ can be found in $O(n\log^2n)$ time.


\section{Preliminary}\label{sec:preliminary} 
When all vertex weights are same, without loss of generality, we assume their weights are one. Let $T^*$ represent an optimal $k$-subtree in shortest path tree(s) of $x^*$ so that $\phi(x^*,G) = \max_{v\in V(T^*)}w_vd(x^*,v)$. Although the optimal solution may not be unique, the following observation helps figure out one of them.

\begin{observation}\label{obs:x^*interiorofT^*}
    There exists an optimal solution where $x^*$ is a point of $T^*$. For the unweighted case, $T^*$ is induced by the $k$ closest vertices of $x^*$ on $G$. 
\end{observation}
\begin{proof}
    Suppose no such optimal solutions exist. Let $x^*$ and $T^*$ be any optimal solution. Consider the closest vertex $v'$ in $V(T^*)$ to $x^*$. In fact, $T^*$ is a $k$-subtree in a shortest path tree of $v'$. Let $T'$ be  the shortest path tree of $x^*$ containing $T^*$. It is because for each vertex $v\in V(T^*)$, the (shortest) path between $v$ and $x^*$ on $T'$ is the concatenation of the subpath between $v$ and $v'$ and the subpath between $v'$ and $x^*$; both subpaths are their shortest paths in $G$. Due to $d(v',x^*)>0$, we have $\max_{v\in V(T^*)}w_vd(x^*,v) > \max_{v\in V(T^*)}w_vd(v',v)$. Further, due to $\phi(v',G)\leq\max_{v\in V(T^*)}w_vd(v',v)$, a contradiction where $\phi(x^*, G) > \phi(v',G)$ occurs. Hence, the first statement is proved. 

    Moreover, for any point $x$ in $G$, $x$ and its $k$ closest vertices in $V$ must induce a connected $k$-subtree in $x$'s shortest path tree(s). This implies the second statement. \qed
\end{proof}

\begin{figure}
    \centering
    \includegraphics[scale=0.7]{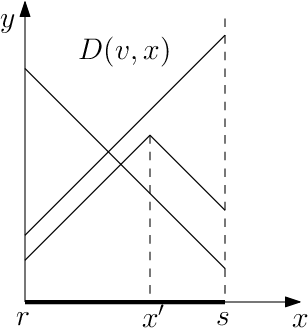}
    \caption{Illustrating the three cases of the (weighted) distance function $D(v,x)$ for $x\in e(r,s)$: As $x$ moves from $r$ to $s$ on $e$, at rate $w_v$, $D(v,x)$ increases, or decreases, or first increases until $v$'s semicircular point $x'$ and then decreases.}
    \label{fig:distancefun}
\end{figure}


Let $e(r,s)$ be an arbitrary edge of $G$. Let $x$ be any point on $e(r,s)$, which is at distance $t(x)$ to $r$ along $e(r,s)$. For each vertex $v\in V$, we use $D(v,x)$ to represent the (weighted) distance from $v$ to $x$. Fig.~\ref{fig:distancefun} shows the three cases of $D(v,x)$. Denote by $I_{e(r,s)}(y,y')$ the path (segment) along $e(r,s)$ between two points $y,y'$ on $e(r,s)$. If there is a point $x'\in e(r,s)$ so that $d(v,r) + t(x') = d(v,s) + l(e(r,s)) -t(x')$, then $v$ has two shortest paths to $x'$ respectively containing $I_{e(r,s)}(r,x')$ and $I_{e(r,s)}(x',s)$. We refer to $x'$ as the semicircular point of $v$ on $e(r,s)$ and also say that $v$ is a \textit{neutral} vertex of $x'$. Notice that every vertex has at most one semicircular point on each edge of $G$. 

Denote by $\bar{V}(x)$ the set of all neutral vertices of a point $x$ of $G$. W.r.t. $x\in e(r,s)$, $V$ can be partitioned into three subsets $\bar{V}(x)$, $V_r(x)$, and $V_s(x)$: $V_r(x)$ is composed of all vertices in $V-\bar{V}(x)$ satisfying the condition $d(v,r) + t(x) < d(v,s) + l(e(r,s)) - t(x)$, and $V_s(x)$ contains all remaining in $V-\bar{V}(x)$ with $d(v,r) + t(x) > d(v,s) + l(e(r,s)) - t(x)$. Indeed, $V_r(x)$ (resp., $V_s(x)$) contains all vertices in $V$ whose shortest paths to $x$ each contains $I_{e(r,s)}(r,x)$ (resp., $I_{e(r,s)}(s,x)$) on $e(r,s)$. 


When $x$ is interior of the edge, we consider $x$ as a dummy vertex in its shortest path tree(s) (but not $G$) except when we say its shortest path tree(s) in $G$. Let the (real or dummy) vertex containing $x$ be the root of its shortest path tree(s). We have the following useful properties.  

\begin{observation}\label{obs:shortestpathtreeproperty}
    In any shortest path tree of $x$, every vertex of $\bar{V}(x)$ is either a leaf or an internal node whose descendants are all in $\bar{V}(x)$. Further, for any two points $y, y'$ on $e(r,s)$ with $t(y)< t(y')$, $\bar{V}(y')\in V_r(y)$ and $V_r(y')\in V_r(y)$, and symmetrically, $\bar{V}(y)\in V_s(y')$ and $V_s(y)\in V_s(y')$. 
\end{observation}
\begin{proof}
    Let $T$ be any shortest path tree of $x$. Suppose a vertex $v\in\bar{V}(x)$ has a descendant $v'$ in $T$ that is not in $\bar{V}(x)$. Due to $v\in \bar{V}(x)$, by the definition, $v'$ must be in $\bar{V}(x)$, which causes a contradiction. This proves the first statement. 
    
    On account of $t(y)<t(y')$, each $v\in\bar{V}(y')$ has all its shortest paths to $y$ containing only $I_{e(r,s)}(r,y)$ of $e(r,s)$ and so does each $v\in V_r(y')$. Hence, $\bar{V}(y')\in V_r(y)$ and $V_r(y')\in V_r(y)$. Likewise, we have $\bar{V}(y)\in V_s(y')$ and $V_s(y)\in V_s(y')$. \qed
\end{proof}
 
For the case $k=n$, as proved in~\cite{ref:KarivAnC79}, $\lambda^*$ is in the set of the (weighted) one-center objective value of every two vertices by constraining their center lying on each edge of $G$, that is, the set of the $y$-coordinates of the intersections between every two functions $D(v,x)$ w.r.t. each edge where their slopes are of opposite signs. Denote by $\Lambda$ this set.  Additionally, let $\Lambda'$ be the set of values $D(v,x)$ of each $v\in V$ at all $O(mn)$ semicircular points in $G$. Clearly, $|\Lambda| = O(mn^2)$ and so is $|\Lambda'|$. For the general case $k\leq n$, the following observation holds. 

\begin{observation}\label{obs:lambda_belongs_twosets}
    $\lambda^*\in \Lambda\cup\Lambda'$. 
\end{observation}
\begin{proof}
    Clearly, it is sufficient to show that if $x^*$ is not at any semicircular point then $\lambda^*\in\Lambda$. On the one hand, $x^*$ is interior of an edge, say $e(r,s)$, in $T^*$. Denote by $x_r$ (resp., $x_s$) the closest semicircular point to $x^*$ in the interval $[r, x^*)$ (resp., $(x^*,s]$) on $e(r,s)$, and if no semicircular points are in $[r, x^*)$ (resp., $(x^*,s]$), then we let $x_r$ (resp., $x_s$) be $r$ (resp., $s$). 
    
    Assume that $\lambda^*$ is not caused by the intersection of two functions $D(v,x)$ where their slopes are of opposite signs, i.e., $\lambda^*\notin\Lambda$. By Observation~\ref{obs:shortestpathtreeproperty}, for all points $x\in (x_r, x_s)$ on $e(r,s)$, $\bar{V}(x)=\emptyset$ and both sets $V_r(x)$ and $V_s(x)$ remain same. This implies that for any point in $(x_r, x_s)$, $T^*$ is a subtree of its one shortest path tree in $G$. In addition, due to the assumption, vertices of $V(T^*)$ whose values $D(v,x^*)$ equal to $\lambda^*$ are all in either $V_r(x^*)$ or $V_s(x^*)$. Hence, there must be a point in $(x_r, x^*)\cup (x^*, x_s)$ such that the maximum (weighted) distance of (vertices in) $V(T^*)$ to this point is smaller than that of $V(T^*)$ to $x^*$, i.e., $\lambda^*$. This means that $\phi(x, G)$ achieves a smaller objective value than $\lambda^*$ at a point in $(x_r, x^*)\cup (x^*, x_s)$ on $e(r,s)$, which leads a contradiction. 

    On the other hand, $x^*$ is at a vertex of $V(T^*)$, say $v'$. Suppose $x^*$ is not a center of any two vertices, i.e., $\lambda^*\notin \Lambda$. Because $x^*$ is not a semicircular point. There must be an edge, say $e'$, in $T^*$ incident to $v'$ with the following properties: Let $x''$ be the semicircular point on $e'$ closest to $v'$ and $x''$ 
    is set as the other incident vertex $v''$ of $e'$ if no semicircular points are in $e'$. For any point in $[x^*, x'')$, $T^*$ is a subtree of its one shortest path tree in $G$, and importantly, by considering $v'$ as the root of $T^*$, vertices in $V(T^*)$ that decide $\lambda^*$ are all in the subtree of $T^*$ rooted at $v''$. Hence, a smaller objective value than $\lambda^*$ can be achieved at any point in $(x^*, x'')$ on $e'$, which leads a contradiction. 
    
    
    Therefore, $\lambda^*\in \Lambda\cup\Lambda'$. \qed
    \end{proof}

For the unweighted case or the (weighted) tree version, we have the following observation. 

\begin{observation}\label{obs:lambda_belongs_oneset}
    For the unweighted case or the (weighted) tree version, $\lambda^*\in \Lambda$, and $x^*$ is the (weighted) center of $T^*$ with respect to $V(T^*)$. 
\end{observation}
\begin{proof}
    We first discuss the unweighted case. Assume $x^*$ is on an edge $e'$ of $T^*$. Suppose $\lambda^*\notin\Lambda$, that is, $x^*$ is not the (local) center of any two vertices of $V(T^*)$ on $e'$. When $x^*$ is interior of $e'$, either all vertices of $V(T^*)$ whose distances to $x^*$ equal to $\lambda^*$ have their semicircular points on $e'$ at $x^*$, or their semicircular points on $e'$ are not at $x^*$ but these vertices are all in the same subtree generated by removing $x^*$ from $T^*$. Although $x^*$ might be a semicircular point, moving $x^*$ along $e'$ toward any incident vertex of $e'$ reduces the maximum distance of $x^*$ to $V(T^*)$ in the former case while moving $x^*$ along $e'$ toward that subtree in the later case does so. It follows that there exists a point on $e'$ whose distance to its $k$-th closest vertex of $V$ is less than $\lambda^*$, which means that a point of $e'$ leads a smaller objective value than $\lambda^*$. Hence, a contradiction occurs. 
    
    
    Otherwise, $x^*$ is at an incident vertex $v'$ of $e'$. Vertices whose distances to $v'$ equal to $\lambda^*$ are all in the same subtree generated by removing $v'$ from $T^*$, so moving $x^*$ toward this subtree along the edge connecting it and $v'$ reduces the maximum distance of $x^*$ to $V(T^*)$. Or these vertices are in different subtrees (generated by removing $v'$ from $T^*$) but $v'$ must have an incident edge $e''$ so that all these vertices excluding those connected with $v'$ by $e''$ have their semicircular points on $e''$ at $v'$. Hence, moving $x^*$ along $e''$ away from $v'$ reduces the maximum distance of $x^*$ to $V(T^*)$, which causes a contradiction. Hence, $\lambda^*\in\Lambda$.
   
    Furthermore, suppose $x^*$ is not the center of $T^*$ w.r.t. $V(T^*)$. Because the maximum distance of $V(T^*)$ to their center on $T^*$ is smaller than $\lambda^*$. The distance from their center to its $k$-th closest vertex in $V$ is smaller than $\lambda^*$. Due to Observation~\ref{obs:x^*interiorofT^*}, their center leads a smaller objective value than $\lambda^*$. Thus, the statement is true for the unweighted case. 

    When $G$ is a tree graph, every vertex has no semicircular points. Hence, $\lambda^*\in\Lambda$. Notice that unlike the graph version, for every pair of vertices, their center on their (unique) path of $G$ is the only intersection between their distance functions $D(v,x)$ where their functions $D(v,x)$ are of slopes with opposite signs, which implies $\Lambda = O(n^2)$. In addition, since the unique shortest path tree of every point is $G$ itself, $\phi(x,G)$ at the center of $T^*$ w.r.t. $V(T^*)$ is not larger than $\max_{v\in V(T^*)}w_vd(v,x^*)$. Thus, the observation holds. \qed
\end{proof}

As in~\cite{ref:DearingAM74}, the (weighted) diameter $W(G)$ of $G$ is defined as $\max_{v,u\in V} \frac{w_vw_ud(u,v)}{w_v + w_u}$; when all weights are same, clearly, $W(G)$ is exactly one half of the diameter of $G$. Regarding to our problem, the observation below reveals the equivalency between our problem and the problem of finding the $k$-subtree of minimum diameter in a graph. 

\begin{observation}\label{obs:smallestdiameterequivalency}
For the unweighted case, $T^*$ is of minimum diameter among all $k$-subtrees of $G$ if and only if $W(T^*) = \lambda^*$. When $G$ is a tree, $T^*$ is a $k$-subtree of minimum diameter. 
\end{observation}
\begin{proof} 
    We first prove the second statement. For $G$ being a tree graph, due to Observation~\ref{obs:lambda_belongs_oneset}, finding $T^*$ is equivalent to finding 
    a $k$-subtree with the smallest optimal one-center objective value among all $k$-subtrees of $G$, that is, finding a $k$-subtree with minimum (weighted) diameter. 

    For $G$ being a general graph whose vertices are of weights one, we define $\Phi(x, G)$ at any point $x\in G$ as $\max_{v\in V} d(v,x)$. Clearly, $\Phi(x, G)$ is the objective value of the one-center problem at $x$, which equals to $\phi(x, G)$ for $k=n$. Note that $W(G)\leq\Phi(x, G)$~\cite{ref:DearingAM74}. Let $G'$ be any subgraph of at least $k$ vertices in $G$. We observe that at any point $x'$ of $G'$, $\phi(x', G')\geq\phi(x', G)$. 
    Because in the unweighted case, $\phi(x', G')$ (resp., $\phi(x', G)$) 
    equals to the distance of $x'$ to its $k$-th closest vertex in $G'$ (resp., $G$). Clearly, this observation also holds for the weighted tree version. With these properties, the first statement is proved as follows. 
    
    Suppose $W(T^*) = \phi(x^*, G) = \lambda^*$. 
    Let $T'$ be any $k$-subtree of $G$, and denote by $c'$ its center. Clearly, $W(T^*) = \phi(x^*, G) \leq \phi(c', G)$. 
    By the above properties, we have $\phi(c',G) \leq \phi(c',T') = \Phi(c',T')$. Since $T'$ is a tree, $\Phi(c',T') = W(T')$. Hence, $W(T^*)\leq W(T')$. 
 
    On the other hand, suppose $T^*$ is of minimum diameter among all $k$-subtrees of $G$. By Observation~\ref{obs:lambda_belongs_oneset}, we have $W(T^*) = \Phi(x^*, T^*)$. Due to $\Phi(x^*, T^*) = \phi(x^*,G)$, $W(T^*) = \phi(x^*, G) =\lambda^*$. \qed
\end{proof}

Furthermore, the following corollary can be utilized to determine whether any given $k$-subtree $T^k$ of $G$ is of minimum diameter among all its $k$-subtrees. 

\begin{corollary}\label{cor:proveTsmallestdiameter}
    For the unweighted case or the weighted tree version, $T^k$ is of minimum diameter if and only if $W(T^k) = \lambda^*$.
\end{corollary}
\begin{proof}
    The properties in the proof of Observation~\ref{obs:smallestdiameterequivalency} supports the following proof for this corollary. Suppose $W(T^k) = \phi(x^*, G)$. Let $T'$ be any $k$-subtree of $G$ and $c'$ be its center. We then show $W(T^k)\leq W(T')$. 
    Clearly, $W(T^k) = \phi(x^*, G) \leq \phi(c', G) \leq \phi(c', T')$. 
    Due to $\phi(c', T') = \Phi(c', T')$, $ W(T^k)\leq\Phi(c', T') = W(T')$. Hence, $T^k$ is of minimum diameter. 
    
    For the other direction, suppose $T^k$ is of minimum diameter. Hence, $W(T^k)\leq W(T^*)$. Let $c''$ be the center of $T^k$. 
    Due to $W(T^*) = \Phi(x^*, T^*)$, $W(T^k)\leq\Phi(x^*, T^*)\leq\phi(x^*, G)\leq \phi(c'', G)$. 
    Moreover, $\phi(c'', G) \leq \phi(c'', T^k) = \Phi(c'', T^k)$. 
    It follows that $W(T^k)\leq \phi(x^*, G)\leq\Phi(c'',T^k) = W(T^k)$. 
    Hence, $W(T^k) = \phi(x^*, G) =\lambda^*$. \qed
\end{proof}

In the following, when we talk about a point $x$ on an edge, we use $x$ to denote $t(x)$ for convenience. 
In addition, for any point $p\in\mathbb{R}^2$, we use $x(p)$ and $y(p)$ to denote its $x$- and $y$-coordinates, respectively; if the context is clear, for a point $x$ on the $x$-axis, we directly use $x$ to denote its $x$-coordinate. 



\section{Solving the Problem on a Vertex-Weighted Graph}\label{sec:generalweightedversion}
In this section, we shall present our algorithm for the weighted version on undirected graphs. To introduce our algorithm, we first give the result for the feasibility test that determines for any given value $\lambda$, whether $\lambda\geq\lambda^*$, and defer its proof in Section~\ref{sec:generalweightedfeasibility}. 

\begin{lemma}\label{lem:generalweightedfeasibility}
Given any value $\lambda$, we can decide whether $\lambda\geq\lambda^*$ in 
$O(mn\log n + m^2\log n)$ time. 
\end{lemma}

Our idea for computing $\lambda^*$ is as follows: First, we compute a set of $O(mn)$ lines in the $x,y$-coordinate plane generated by extending line segments on graphs of functions $y=D(v,x)$ of all vertices w.r.t. every edge of $G$ and including vertical lines through incident vertices of every edge on $x$-axis. Clearly, the set $\Lambda\cup \Lambda'$ belongs to the set of $y$-coordinates of all intersections between these obtained lines, and $\lambda^*$ is the $y$-coordinate of the lowest intersection with a feasible $y$-coordinate. Next, we compute $\lambda^*$ by finding that lowest one among all intersections by utilizing the line arrangement search technique~\cite{ref:ChenAn13} with the assistance of our feasibility test in Lemma~\ref{lem:generalweightedfeasibility}. The line arrangement search technique is reviewed as follows. 

Suppose $L$ is a set of $N$ lines in the plane. Denote by $\calA(L)$ the arrangement of lines in $L$. In $\calA(L)$, every intersection of lines defines a vertex of $\calA(L)$ and vice versa. Let $v_1(L)$ be the lowest vertex of $\calA(L)$ whose $y$-coordinate $y(v_1(L))$ is a feasible value, and let $v_2(L)$ be the highest vertex of $\calA(L)$ whose $y$-coordinate $y(v_2(L))$ is smaller 
than $y(v_1(L))$. By the definitions, $y(v_2(L))<\lambda^*\leq y(v_1(L))$ and no vertices in $\calA(L)$ have $y$-coordinates in range $(y(v_2(L)), y(v_1(L)))$. 
Lemma~\ref{lem:linearrangementtechnique} was given in~\cite{ref:ChenAn13} to find the two vertices. 

\begin{lemma}\label{lem:linearrangementtechnique}
     \textup{\cite{ref:ChenAn13}} Both vertices $v_1(L)$ and $v_2(L)$ can be computed in $O((N+\tau)\log N)$ time, where $\tau$ is the running time of the feasibility test. 
\end{lemma}

Regarding to our problem, if these $O(mn)$ lines are known, with the assistance of Lemma~\ref{lem:generalweightedfeasibility}, we can adapt Lemma~\ref{lem:linearrangementtechnique} to compute $\lambda^*$ in $O(mn\log n\log mn + m^2\log n\log mn)$ time. 

It remains now to compute these $O(mn)$ lines. We use a list $L$ to store all these lines, which is empty initially. For every edge of $G$, we perform the following $O(n)$-time routine: Suppose we are about to process edge $e(r,s)$. Consider $e(r,s)$ being on the $x$-axis with vertex $r$ at the origin and vertex $s$ at the point of $x$-coordinate $l(e(r,s))$. First, we join into $L$ the two vertical lines respectively through vertices $r,s$ on $x$-axis. Next, for each $v\in V$, we determine in $O(1)$ time function $y = D(v,x)$ w.r.t. $x\in e(r,s)$ with the provided distance matrix, and then insert into $L$ the lines containing all $O(1)$ line segments on its graph. 

As a result, a set of $O(mn)$ lines stored in $L$ is obtained in $O(mn)$ time. Clearly, the set $\Lambda\cup \Lambda'$ belongs to the set of $y$-coordinates of intersections between lines in $L$. Now we can employ Lemma~\ref{lem:linearrangementtechnique} to find the lowest vertex with a feasible $y$-coordinate in the line arrangement $\calA(L)$ by applying Lemma~\ref{lem:generalweightedfeasibility} to decide the feasibility of every tested $y$-coordinate. The $y$-coordinate of that lowest vertex is exactly $\lambda^*$. Accordingly, $x^*$ and $T^*$ can be found by applying Lemma~\ref{lem:generalweightedfeasibility} to $\lambda^*$.  
Thus, we have the following theorem. 

\begin{theorem}\label{the:generalweightedresult}
    The weighted connected $k$-vertex one-center problem can be solved in $O(mn\log n\log mn + m^2\log n\log mn)$ time. 
\end{theorem}

\subsection{The Feasibility Test}\label{sec:generalweightedfeasibility}
For any given value $\lambda$, we say that a subgraph $G'$ of $G$ can be covered by a point $x$ in $G$ (under $\lambda$) if and only if the maximum (weighted) distance of $x$ to $V(G')$ is no more than $\lambda$. Indeed, the feasibility test asks for the existence of a point in $G$ that covers a $k$-subtree of its shortest path tree(s) in $G$ (without joining a dummy vertex for $x$ into its shortest path trees). For any point $x$ of $G$, there is a subtree of maximum cardinality in $G$ that is covered by $x$ and whose vertices and $x$ induce a connected subtree in a shortest path tree of $x$; this subtree either contains $x$ inside or excludes it but has a vertex adjacent to $x$ in that shortest path tree. (The latter situation may occur only if $x$ is interior of an edge.) Refer to this subtree as the \textit{largest covered self-inclusive subtree} of $x$, and denote it by $T_\lambda(x)$. The following key observation leads our decision algorithm. 

\begin{observation}\label{obs:decision}
    There exists a point $x'$ in $G$ with $|V(T_\lambda(x'))|\geq k$ if and only if $\lambda$ is feasible. 
\end{observation}
\begin{proof}
    It suffices to show that if $\lambda\geq\lambda^*$ then such a point must exist. Since $\lambda$ is feasible, there must be a point on $G$ so that it covers a $k$-subtree of its shortest path trees in $G$. Let $x'$ be such a point. Clearly, the statement is true if a $k$-subtree of its one shortest path tree covered by $x'$ contains $x'$ or has a vertex adjacent to $x'$ in that shortest path tree. 
    
    Otherwise, the largest self-inclusive subtree of its every shortest path tree in $G$ covered by $x'$ is of size less than $k$. The largest subtree covered by $x'$ in $G$ contains at least $k$ vertices of $V$ but is not adjacent to $x'$ in its shortest path trees. Let $T^k$ be such a largest subtree covered by $x'$. Suppose $T^k$ is rooted at vertex $u$ in its shortest path tree including $T^k$. (Recall that each shortest path tree of $x'$ is rooted at the dummy or real vertex containing $x'$.) Because for each vertex $v\in V(T^k)$, the path of $v$ to $u$ on $T^k$ is exactly its shortest path to $u$ on $G$. $T^k$ is thus a subtree of a shortest path tree of $u$ in $G$; additionally, it includes $u$ and is covered by $u$ under $\lambda$. Due to $|V(T^k)|\geq k$, the observation holds. \qed
\end{proof}


The underlying idea of our feasibility test is: For every edge $e$ of $G$, we determine $|V(T_\lambda(x))|$, i.e., the size $|T_\lambda(x)|$ of $T_\lambda(x)$, for points on $e$ where $|V(T_\lambda(x))|$ changes. In the process, if a point is found so that its $V(T_\lambda(x))$ is of size at least $k$, then $\lambda$ is feasible and so we immediately return. Otherwise, no such points exist and thus $\lambda<\lambda^*$. 

Let $S$ be a subset of $V$. For any point $x\in G$, we refer to a vertex as a \textit{descendant} of subset $S$ 
w.r.t. $x$ if its every shortest path to $x$ contains vertices in $S$. We say a vertex is a \textit{heavy} vertex of $x$ if it cannot be covered by $x$. Clearly, $T_\lambda(x)$ includes neither any heavy vertex of $x$ nor any descendant of the set $H(x)$ of all its heavy vertices. Otherwise, a vertex is called a \textit{light} vertex of $x$ (if it is neither a vertex in $H(x)$ nor a descendant of $H(x)$). 

The following observation sets the base for determining $|V(T_\lambda(x))|$. 

\begin{observation}\label{obs:computinglargesttree}
    $T_\lambda(x)$ is induced by all light vertices of $x$. 
\end{observation}
\begin{proof}
    It is sufficient to show that $x$ must have a shortest path tree where the path from $x$ to its every light vertex contains no (heavy) vertex in $H(x)$. 
    
    On the one hand, every vertex adjacent to $x$ in $G$ is in $H(x)$. It follows that if $x$ is interior of an edge in $G$ then $x$ has no light vertices, and otherwise, the one containing $x$ is its only light vertex. Hence, the statement is true in this situation. 
   
    On the other hand, $x$ is adjacent to at least one of its light vertices. Let $G'$ be the subgraph generated by removing $H(x)$ and all descendants of $H(x)$ from $G$. Note that if $x$ is interior of an edge and only one of its two adjacent vertices is light, then $x$ is not in $G'$, which contains its only light adjacent vertex though; otherwise, $x$ is in $G'$. For the former case, to maintain the reachability, we join a dummy vertex for $x$ into $G'$ by connecting $x$ and its only light adjacent vertex with an edge of the same length as their segment length along the edge of $G$ containing $x$.     
    
    By the definition, every light vertex of $x$ has at least one shortest path to $x$ on $G$ where every vertex is light w.r.t. $x$. This implies the following properties: (1) $G'$ is connected; (2) $G'$ contains such shortest path(s) in $G$ from $x$ to its every light vertex. Hence, there must be a shortest path tree w.r.t. $x$ where the path from $x$ to its every light vertex contains only its light vertices. 

    Thus, the observation holds.\qed
\end{proof}

Furthermore, a problem needs to be addressed for determining $|V(T_\lambda(x))|$: 
Given any subset $S$ of $V$, the goal is to find all descendants of $S$ 
w.r.t. $x$ on $G$ and the subgraph generated by removing $S$ 
and its descendants from $G$. The following lemma gives the result. 

\begin{lemma}\label{lem:finddescendants}
    With $O(n+m\log n)$-time preprocessing work, given any subset $S$ of $V$, all descendants of $S$ w.r.t. $x$ and the subgraph generated by removing $S$ and its descendants from $G$ can be obtained in $O(n' + m'\log n)$ time where $n'$ is the total number of vertices in $S$ and its descendants, and $m'$ is their total degrees. 
\end{lemma}
\begin{proof}
    The computations will be performed on a copy $G'$ of $G$. First, join a dummy vertex of weight zero into $G'$ for point $x$ if $x$ is interior of an edge. Let $v_x$ be the dummy or real vertex containing $x$. The vertex set $V'$ of $G'$ is $V\cup\{v_x\}$, and $v_x\notin S$ if $v_x$ is a dummy vertex. 
    Clearly, if $v_x\in S$, which can be verified in $O(n)$ time, then every vertex in $V-S$ is a descendant of $S$ and the subgraph generated by removing $S$ and its descendants from $G$ is a null graph. In general, $v_x\notin S$ but it is clear to see that every descendant of $S$ in $G'$ is a descendant of $S$ in $G$ and vice versa. 
    
    To find all descendants of $S$ w.r.t. $v_x$, in the preprocessing work, we first process every vertex in $G'$ by utilizing any self-balanced binary search tree to store all its adjacent vertices instead of a list, which can be carried out in $O(n+m\log n)$ time. With such representation of $G'$, for any vertex of $V'$, its `adjacency' binary search tree is still of size linear to its degree but deleting an adjacent vertex takes $O(\log n)$ time. As will be shown later, this representation improves the efficiency on finding the subgraph induced by $S$ and its descendants.  
    
    Next, we compute for 
    each vertex $v$ the following two sets: set $P$ consists of the preceding vertices of $v$ 
    on all its shortest paths to $v_x$, and set $C$ is composed of all vertices 
    whose sets $P$ include $v$, i.e., its succeeding vertices on all shortest paths to $v_x$. For $v_x$, let its set $P$ be empty. Because sets $P$ and $C$ of every vertex each contains at most all its adjacent vertices. The total size of sets $P$ and $C$ for all vertices is no more than $2(m+2)+n+1$. 

    
    Since every edge has a positive length, we can compute the sets $P$ and $C$ for all vertices of $V'$ as follows. For each vertex $v$ of $V'$, we determine for its each adjacent vertex $u$ whether $d(x,v) + l(e(v,u)) = d(x,u)$; if yes then $v$ is the preceding vertex of $u$ on $u$'s one shortest path to $v_x$, so we insert $v$ into $u$'s set $P$ and insert $u$ to $v$'s set $C$. Last, for $v_x$, we set its set $P$ as empty. To achieve a logarithmic-time insertion and deletion, we utilize any self-balanced binary search tree to maintain the two sets for each vertex. Since the distance of every vertex to $x$ is known in $O(1)$ time, the total running time is $O(n+m\log n)$. As a result, the total preprocessing time is $O(n+m\log n)$. 
    


    The reason why we need these sets $P$ and $C$ for finding $S$'s descendants is explained as follows. For every descendant of $S$, on its each shortest path to $v_x$, every intermediate vertex of the subpath from it to its closest vertex in $S$ must be a descendant of $S$; because otherwise, it has a shortest path to $v_x$ excluding any vertex of $S$. This means that the set $P$ of any vertex $v\in V'-S$ only contains $S$'s vertices and its descendants iff $v$ is a descendant of $S$. Additionally, for every descendant of $S$, consider the link number of the above subpath for its each shortest path, and refer to their maximum as the \textit{shortest-path depth} of this descendant to $S$; define the shortest-path depth of each vertex in $S$ to $S$ as zero. Clearly, all $S$'s descendants can be found level by level in the ascending order of their shortest-path depths to $S$ by iteratively removing $S$ and its descendants found following that order from sets $P$ of remaining vertices such that if the set $P$ of a vertex becomes empty, then this vertex is a new descendant of $S$ whose shortest-path depth to $S$ is smallest among all unknown (remaining) vertices but not less than that of any known descendant. 

    
    
    Now we are ready to show how to find all descendants of $S$ in order. We use a list $L$ to store $S$ and all descendants of $S$. Initially, we insert each vertex of $S$ into $L$. Next, create a queue $Q$ to maintain $S$ and $S$'s descendants found during the procedure so that all vertices in $Q$ are in that ascending order and need to be removed from sets $P$ of vertices in $V'-L$ further. Then we insert each vertex $v\in S$ into $Q$ after breaking the connection between $v$ and its preceding vertices on all shortest path to $v_x$ as follows: Delete $v$ from set $C$ of each vertex in $v$'s set $P$ and then emptify $v$'s set $P$. At the beginning, $Q$ thus contains all vertices of the shortest-path depths to $S$ being zero. 
    

    
    We proceed to find $S$'s descendants (level by level) in the ascending order by their shortest-path depths to $S$. Every step we extract the front vertex $v$ from $Q$, which is of the smallest shortest-path depth in $Q$, and process every vertex in its set $C$ as follows. For each vertex $u$ in $v$'s set $C$, we remove $v$ from $u$'s set $P$ and delete $u$ from $v$'s set $C$; then we check whether $u$'s set $P$ is empty or not. If yes, then $u$ is a descendant of $S$ (since only $S$ and $S$'s descendants are deleted from sets $P$), so we insert $u$ into $Q$ and add it to list $L$. Note that $u$ must be of the smallest shortest-path depth among all remaining descendants in $V'-L$ since every descendant with a smaller shortest-path depth has been discovered. 
    
    It also should be mentioned that when a vertex is inserted into $Q$, not only its set $P$ is empty but also it is not in any set $C$. Because every newly found descendant is only in sets $C$ of vertices in $S$ or its descendants that have been found. 

    We terminate this procedure once $Q$ becomes empty. At this moment, 
    all descendants of $S$ in $G'$ are found and stored in list $L$, 
    and their information have been removed from sets $C$ and $P$ of all remaining vertices in $G'$. Through the whole procedure, 
    vertices of $S$ and $S$'s descendants each is explored exactly once, 
    and for each vertex in their sets $C$ and $P$, at most two deletion operations are executed, which take $O(\log n)$ time. 
    Additionally, the total size of sets $C$ and $P$ of every vertex is 
    bounded by the degree of this vertex in $G'$. 
    As a result, all descendants of $S$ in $G$ can be found in $O(n'+m'\log n)$ time. 

    To generate the subgraph by removing $S$ and its descendants from $G$, during the above procedure, 
    once we finish exploring a vertex $v$ in $Q$, 
    we delete $v$ from $G'$ as follows: For its each adjacent vertex $u$, we delete $v$ from the adjacency binary search tree of $u$ in $O(\log n)$ time. Subsequently, we remove $v$ (and its adjacent information) from $V'$ in $O(1)$ time. Last, if $v_x$ is a dummy vertex then we delete $v_x$ as the above from $G'$, which can be done in $O(\log n)$ time since its degree is two. Clearly, the total running time remains $O(n'+m'\log n)$. 
 
    
    In a sum, with the $O(n + m\log n)$-time preprocessing work, 
    all descendants of $S$ and the subgraph of $G$ induced by $V$ excluding $S$ and all its descendants can be found in $O(n' + m'\log n)$ time.  \qed
\end{proof}

We now present our algorithm for determining values $|V(T_\lambda(x))|$ 
at necessary points on an arbitrary edge $e(r,s)$. We simply use $e$ to denote edge $e(r,s)$. For any point $x\in e$, we define $f(e,x) = |V(T_\lambda(x))|$. Recall that $x$ is at distance $t(x)$ to $r$ along $e$. $f(e,x)$ is indeed a function w.r.t. $t(x)$. The following Lemma can be employed to construct $f(e,x)$.     

\begin{lemma}\label{lem:lightvertexfunction}
    $f(e,x)$ is a piecewise constant function of complexity $O(n)$. The ordered set of all its breakpoints can be computed in $O((m+n)\log n)$ time. 
\end{lemma}
\begin{proof} 
    Recall that $\bar{V}(x)$ is the set of neutral vertices of $x$ in $V$, and $V_r(x)$ (resp., $V_s(x)$) is the set of all vertices in $V$ whose shortest paths to $x$ each contains the segment $I_e(r,x)$ (resp., $I_e(x,s)$) on $e$ between $r$ (resp., $s$) and $x$. Decomposing any shortest path tree of $x$ at the dummy or real vertex $v_x$ containing $x$ generates two subtrees $T_r$ and $T_s$ rooted at $v_x$ so that $T_r$ including $I_e(r,x)$ is induced by $v_x$, the set $V_r(x)$, and a subset of $\bar{V}(x)$, and $T_s$ including $I_e(x,s)$ is induced by $v_x$, the set $V_s(x)$ and the remaining in $\bar{V}(x)$. Note that if $v_x$ is $r$ then $v_x\in V_r(x)$ and $v_x\notin V_s(x)$, and if $v_x$ is $s$ then $v_x\in V_s(x)$ and $v_x\notin V_r(x)$; otherwise, $v_x$ is a dummy vertex and thus belongs to neither $V_r(x)$ nor $V_s(x)$. 
    
    
    For any vertex $v\in V$, we say $v$ is light w.r.t. $x$ by $r$ (resp., $s$) if it has a shortest path to $x$ that contains $I_e(r,x)$ (resp., $I_e(x,s)$) but excludes any heavy vertex of $x$. Clearly, every light vertex of $x$ in $V$ is by either $r$ or $s$. Only vertices in $\bar{V}(x)$ might be light w.r.t. $x$ by both $r$ and $s$; whereas every light vertex of $x$ in $V_r(x)$ (resp., $V_s(x)$) is by only $r$ (resp., only $s$).

    Let $Q_r(x)$ (resp., $Q_s(x)$) be the subset of $V$ that contains $x$'s light vertices in $V_r(x)$ (resp., $V_s(x)$) and all its light vertices in $\bar{V}(x)$ by $r$ (resp., $s$). Clearly, $x$ must have a shortest path tree such that $Q_r(x)$ induces a subtree rooted at $r$ in $T_r$; it must have a shortest path tree, which might be different to the one for $Q_r(x)$, such that $Q_s(x)$ induces a subtree rooted at $s$ in $T_s$. Let $Q_{rs}(x)$ be subset $Q_r(x)\cap Q_s(x)$. Clearly, $Q_{rs}(x) = \emptyset$ if $x$ is not a semicircular point. By Observation~\ref{obs:computinglargesttree}, we have at any $x\in e$, $f(e,x) = |Q_r(x)| + |Q_s(x)| - |Q_{rs}(x)|$. 

    Define $f_r(e,x) =|Q_r(x)|$, $f_s(e,x)= |Q_s(x)|$, and $f_{rs}(e,x) = |Q_{rs}(x)|$ for $x\in e$. Below, we shall present several properties of $Q_r(x)$ and the algorithm for determining $f_r(e,x)$ guided by these properties. Since $f_s(e,x)$ can be derived in a similar way, we omit the details of computing $f_s(e,x)$. While determining $f_r(e,x)$ and $f_s(e,x)$, by maintaining for every semicircular point on $e$ the \textit{coverage} information that indicates whether its neutral vertices each is light by $r$ or $s$, $Q_{rs}(x)$ at each semicircular point can be obtained in the time linear to the number of its neutral vertices. 

    For any two points $y<y'$ on $e$ (i.e., $t(y)<t(y')$), by Observation~\ref{obs:shortestpathtreeproperty}, we have $V_r(y')\cup\bar{V}(y')\in V_r(y)$. Hence, $Q_r(y')\in V_r(y)$. Additionally, every vertex in $Q_r(y')$ must be light w.r.t. $y$ by $r$. Hence, $Q_r(y')\in Q_r(y)$ and $f_r(e,y')\leq f_r(e,y)$. 
    
    Furthermore, for any vertex $v$ whose function $y = D(v,x)$ is a line segment of slope $+w_v$ for $x\in e$ (i.e., $t(x)\in [0,l(e)]$) in the $x,y$-coordinate plane, solving $D(v,x)=\lambda$ generally generates a point $x_v$ on $e$ so that $v$ becomes heavy w.r.t. any $x\in (x_v, s]$; if $D(v,x)<\lambda$ at any $x\in e$, let $x_v$ be $s$ and thereby $(x_v,s]$ contains no points on $e$; if $D(v,x)>\lambda$ at any $x\in e$, let $x_v$ be any adjacent vertex of $r$ but not $s$, and $(x_v,s]$ contains every point of $e$. Clearly, $v$ is heavy w.r.t. any $x\in (x_v, s]$. For any vertex $v$ that has a semicircular point on $e$, solving $D(v,x) =\lambda$ leads at most two points on $e$, and let $x_v$ be the closest one along $e$ to $r$ among them and $v$'s semicircular point; if $D(v,r) >\lambda$, similarly, let $x_v$ be any adjacent vertex of $r$ but not $s$. Clearly, $v$ is not light at any $x\in (x_v, s]$ by $r$ (but might be light by $s$). Refer to such point $x_v$ as the \textit{turning} point of $v$ on $e$ w.r.t. $r$.  

    Let $V_r$ be the set including all vertices whose functions $y=D(v,x)$ are of slope $+w_v$ for $x\in e$ and all vertices that have semicircular points on $e$. 
    Clearly, $V_r$ is exactly the set $\cup_{x\in e}\{V_r(x)\cup\bar{V}(x)\}$. The subgraph $G_r$ of $G$ induced by $V_r$ is connected since every vertex of $V_r$ has a shortest path to $r$ containing only $I_e(r,r)$ of $e$.

    Let $S_r = \{x_1, x_2, \cdots, x_z\}$ be the ordered distinct set of the turning points on $e$ of vertices in $V_r$ w.r.t. $r$ such that every point belongs to $e$ and they are sorted ascendingly by their distances to $r$ along $e$. So, $z\leq n$. By Observation~\ref{obs:shortestpathtreeproperty}, for each $1\leq i\leq z$, at any $x\in (x_i,x_{i+1})$, $V_r(x) = V_r(x_{i}) = V_r(x_{i+1})\cup\bar{V}(x_{i+1})$ due to $\bar{V}(x)=\emptyset$. Additionally, each $v\in V_r$ is not in $Q_r(x)$ at any $x\in (x_v, s]$. Hence, $Q_r(x_1) = Q_r(r)$ and for each $1\leq i\leq z$, $Q_r(x) = Q_r(x_{i+1})$ at any $x\in (x_i,x_{i+1}]$. 
    
    Due to $Q_r(x_{i+1})\in Q_r(x_i)$ for each $1\leq i< z$, $f_r(e,x)$ is a piecewise constant function w.r.t. $x\in e$ (i.e., $t(x)\in [0, l(e)]$), and it breaks and falls at only points in $S_r$. See Fig.~\ref{fig:generalf(x)} for an example. Symmetrically, $f_s(e,x)$ is a piecewise constant function of complexity $O(n)$ that monotonically decreases as $x$ moves away from $s$ along $e$ but falls only at points of $e$ where functions $D(v,x)$ break or increase to $\lambda$. 

\begin{figure}[h]
    \centering
    \includegraphics[scale=0.7]{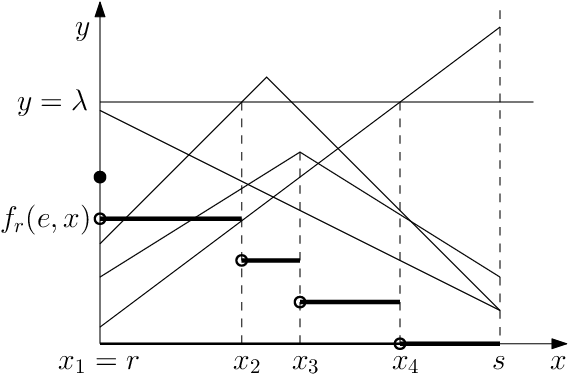}
    \caption{Illustrating that function $f_r(e,x)$ (the heavy line segments) 
    is a piecewise constant function in $x\in e$ that breaks and falls at points where functions $D(v,x)$ increase up to $\lambda$, e.g., $x_2$ and $x_4$, or semicircular points, e.g., $x_1$ and $x_3$.}
    \label{fig:generalf(x)}
\end{figure}
    
    We are now ready to present how to determine $f_r(e,x)$. Start by the following preprocessing work. First, we compute the set $V_r$ and its turning points w.r.t. $r$ on $e$ by adapting the following properties. A vertex $v\in V$ has its function $y=D(v,x)$ of slope $+w_v$ for $x\in [0,l(e)]$ iff $d(v,s) = d(v,r) +l(e)$; $v$ has a semicircular point interior of $e$ if $d(v,s) < d(v,r) +l(e)$ and $d(v,r) < d(v,s) +l(e)$; its semicircular point is at $r$ if $d(v,r) = d(v,s) +l(e)$ and $v$ is not a descendant of subset $\{s\}$ w.r.t. $r$; whereas its semicircular point is at $s$ if $d(v,s) = d(v,r) +l(e)$ and $v$ is not a descendant of subset $\{r\}$ w.r.t. $s$. With these properties and Lemma~\ref{lem:finddescendants}, by the given distance matrix, set $V_r$, all semicircular points, and the turning points of its vertices on $e$ w.r.t. $r$ can be obtained in $O(n + m\log n)$ time.  
    
    Next, we find the subgraph $G_r$ in $G$ induced by $V_r$ in $O(m+n)$ time. As in Lemma~\ref{lem:finddescendants}, we utilize any self-balanced binary search tree to store the adjacent vertices of every vertex in $G_r$; further, we apply that algorithm in Lemma~\ref{lem:finddescendants} to build sets $P$ and $C$ for each vertex in $G_r$ w.r.t. the source vertex $r$. These can be done in $O(n+m\log n)$ time. 
   
    Third, we compute the set $S_r$ by sorting all obtained turning points on $e$ and deleting the duplicates. Simultaneously, we compute for each distinct point $x'$ in $S_r$ two lists $N$ and $H$: $N$ contains all neutral vertices of $x'$ and $H$ contains all non-neutral vertices in $V_r$ whose turning points w.r.t. $r$ are $x'$ on $e$. So, set $N\cup H$ of a point in $S_r$ contains all vertices of $V_r$ whose turning points are $x'$. Clearly, all these operations can be 
    carried out in $O(n\log n)$ time. Last, we insert point $x_0 = (r, s, t(x_0) = -1)$ at the beginning of $S_r$, and set its set $N =\emptyset$ but include in its set $H$ all heavy vertices of $r$ in $V_r$, which can be found in $O(n)$ time.

 
    Last, for each $v\in V_r$, we attach a flag $c_r$ to it so that $c_r$ being true means that $v$ is a light vertex of its semicircular point on $e$ by $r$. Initialize this flag of every vertex in $V_r$ as true. 
    
    Overall, the total preprocessing time is $O((m+n)\log n)$. 

    We shall now determine $f_r(e,x)$, i.e., $|Q_r(x)|$, at each point of $S_r$ in order. For each $0\leq i\leq z$, denote by $G_i$ the subgraph of $G$ that is induced by $Q_r(x_i)$, and $G_i$ must be connected. 
    Define $Q_r(x_0)$ as $V_r$, so $f_r(e,x_0) = |V_r|$ and $G_0 = G_r$. 
    To determine $|Q_r(x_1)|$ (and $G_1$), clearly, it suffices to remove from $G_0$ $x_1$'s heavy vertices and $x_0$'s neutral vertices, i.e., the set $N\cup H$ of $x_0$, as well as all descendants of their set w.r.t. $r$. In general, for each $1\leq i\leq z$, implied by the properties of $Q_r(x)$, $G_i$ and $Q_r(x_i)$ can be obtained by removing 
    from subgraph $G_{i-1}$ all vertices in set $N\cup H$ of $x_{i-1}$ and all descendants of this set in $G_{i-1}$ w.r.t. $r$. 

    Since sets $P$ and $C$ of vertices in $V_r$ w.r.t. $r$ are precomputed, Lemma~\ref{lem:finddescendants} (the main procedure) can be employed to compute $G_i$ for each $1\leq i\leq z$ by iteratively removing from $G_{r}$ all vertices in the set $N\cup H$ of each $x_i$ and all descendants of this set w.r.t. $r$, during which the size of each $G_i$ is attached with $x_i\in S_r$. Clearly, the running time is $O(n + m\log n)$. 

    



    Furthermore, to maintain the coverage information for each vertex in $V_r$ whether it is a light vertex 
    of its semicircular point on $e$ by $r$, additional operations are performed for each $x_i\in S_r$ in the above procedure: For every vertex to remove from $G_i$, we set its flag $c_r$ as false if its semicircular point on $e$ is in segment $(x_i, s]$ on $e$, which can be known in $O(1)$ time.
    
    In summary, with the $O((m+n)\log n)$ preprocessing work, the ordered set of all breakpoints of function $f_r(e,x)$ for $x\in e$ can be obtained 
    in $O(n+m\log n)$ time. 
    
    In an analogous manner, $f_s(e,x)$ for $x\in e$, i.e., the ordered set of all its breakpoints in the ascending order by their distances to $s$, can be determined in $O((m+n)\log n)$ time. Denote by $S_s$ this ordered set of all its breakpoints. While computing $S_s$, a flag $c_s$ is maintained for each vertex to indicate whether it is a light vertex of its semicircular point by $s$. 

    Recall that our goal is to determine $f(e,x)$ that counts the total light vertices of each point $x\in e$, and $f(e,x) = f_r(e,x) + f_s(e,x) - f_{rs}(e,x)$, 
    where $f_{rs}(e,x) =|Q_{rs}(x)|$ and $Q_{rs}(x) = Q_{r}(x)\cap Q_{s}(x)$. Since $Q_{rs}(x) = \emptyset$ at any point on $e$ that is not a semicircular point, $f(e,x)$ is piece-wise constant and breaks at only points in $S_r\cup S_s$ on $e$, which proves the first statement. 
    
    Below, while merging $S_r$ and $S_s$ together into a whole sequence $S$ of distinct points in the ascending order by their distances to $r$ along $e$, we shall determine $f_{rs}(e,x)$ and $f(e,x)$ for every point of $S$ such that $S$ is the ordered set of all breakpoints of $f(e,x)$. All points in $S_r$ (resp., $S_s$) are in ascending order by their distances to $r$ (resp., $s$) along $e$. So, we loop through $S_r$ forward by using index $i$ but loop through $S_s$ backward by using index $j$. 
    
    Every step we compare $x_i\in S_r$ and $x_j\in S_s$ to decide the next point $x_k$ of $S$ in order and compute $f(e,x)$ for $x\in (x_{k-1}, x_k]$. 
    The three cases of $x_i<x_j$, $x_i>x_j$, and $x_i=x_j$ are handled as follows. 

    \begin{enumerate}
        \item Case $x_i<x_j$. Since $x_i$ is closer to $r$ along $e$, $x_k=x_i$. Due to $x_{j+1}<x_i<x_j$, $x_i$ must not be a semicircular point. Hence, $f_{rs}(e,x_i) = 0$. 
        Because $f_s(e,x) = f_s(e,x_{j+1})$ 
        at any $x\in [x_{j+1}, x_j)$ and $f_r(e,x) = f_r(e,x_i)$ at any $x\in (x_{i-1}, x_i]$. 
        For any $x\in (x_{k-1}, x_k]$ in $S$, we have $f(e,x) = f_s(e,x_{j+1}) + f_r(e,x_i)$, 
        which can obtained in $O(1)$ time. 
        So, we join $x_k$ into $S$, attach $f(e,x_k)$ to $x_k$, and increment $i$. 

        \item Case $x_i>x_j$. Set $x_k=x_j$. Due to $x_{i-1}<x_j<x_i$, 
        $x_j$ is not a semicircular point and thereby $f_{rs}(e,x_j) = 0$. 
        Different to the above case, since $f_s(e,x)$ jumps at $x=x_j$, $f(e,x) = f_s(e,x_{j+1}) + f_r(e,x_i)$ for $x\in (x_{k-1}, x_k)$ but at $x=x_k$, $f(e,x) = f_s(e,x_j) + f_r(e,x_i)$. 
        Hence, we cannot attach only $f(e,x_k)$ to $x_k$ (since it means $f(e,x) = f(e, x_k)$ for $x\in (x_{k-1}, x_k]$). 
        To handle such a breakpoint, we attach both above values of $f(e,x)$ to $x_k$ in $S$. 
        Last, $j$ is decremented. 

        \item Case $x_i=x_j$. Set $x_k=x_i$. We first decide in $O(1)$ time 
        whether $x_i$ is a semicircular point on $e$ by checking 
        if its set $N$ is empty or not. Suppose it is not. Neither is $x_j$. So, $f_{rs}(e,x_j) = 0$. 
        As in the case $x_j<x_i$, since $f_s(e,x)$ jumps at $x=x_j$, $f(e,x) = f_s(e,x_{j+1}) + f_r(e,x_i)$ for $x\in (x_{k-1}, x_k)$ 
        but at $x=x_k$, $f(e,x) = f_s(e,x_j) + f_r(e,x_i)$. 
        We attach both values to $x_k$ in $S$, and then increment $i$ but decrement $j$. 

        Otherwise, $x_i$ is a semicircular point on $e$. So, $f_{rs}(e,x_i)\geq 0$. 
        Recall that every vertex $v$ has two flags $c_r$ and $c_s$ 
        where $c_r$ (resp., $c_s$) is true iff $v$ is a light vertex of 
        its semicircular point on $e$ by $r$ (resp., $s$). 
        To compute $f_{rs}(e,x_i)$, one needs to count the number 
        of vertices in the set $N$ of $x_i$ or $x_j$ whose flags 
        $c_r$ and $c_s$ both are true, which can be done in the time 
        linear to the size of $N$ of $x_i$ or $x_j$. Then, we compute in $O(1)$ time the value $f_s(e,x_{j+1}) + f_r(e,x_i)$ for $f(e,x)$ at any $x\in (x_{k-1}, x_{k})$, and value $f_s(e,x_j) + f_r(e,x_i) 
        - f_{rs}(e,x_i)$ for $f(e,x)$ at $x=x_k$. 
        Join $x_k = x_i$ into $S$ and attach both values to it. Last, increment $i$ and then decrement $j$.
    \end{enumerate}

    Clearly, the above merging step takes $O(n)$ time. Combining all above efforts, we can determine in $O((m+n)\log n)$ time $f(e,x)$ (i.e., all its breakpoints) for $x\in e$. \qed
    \end{proof}

    Recall that the feasibility test is to decide if a point exists in $G$ so that its largest covered self-inclusive $T_\lambda(x)$ is of size no less than $k$, i.e., it has at least $k$ light vertices. We can decide the existence of such a point by applying Lemma~\ref{lem:lightvertexfunction} to every edge of $G$. During the procedure, once such a point is found (so that $f(e,x)$ at this point is at least $k$), $\lambda$ is known to be feasible and so we immediately return. Otherwise, no such points exist and hence $\lambda$ is infeasible. As a result, the feasibility of any given value $\lambda$ can be known in $O(m^2\log n + mn\log n)$ time. 

\section{Solving the Problem on a Vertex-Unweighted Graph}\label{sec:generalunweightedversion}
In this section, we introduce the algorithm for the unweighted case where every vertex is of weight one. 

For any edge $e$ of $G$, the point on $e$ minimizing $\phi(x,G)$ among its all points is called the \textit{local} partial center on $e$; denote it by $x^*_e$. $x^*_e$ may not be unique on $e$ but let $x^*_e$ represent any of them. To find $x^*$ on $G$, our strategy is to compute for every edge of $G$ its local partial center $x^*_e$ and the objective value at $x^*_e$ such that $x^*$ is the one with the smallest objective value among them. 

Consider the problem of computing the local partial center on an arbitrary edge $e(r,s)$ of $G$, which is denoted by $e$ for simplicity. As analyzed in Section~\ref{sec:preliminary}, in the unweighted case, at any point $x$ on $G$, $\phi(x,G)$ equals to the distance of $x$ to its $k$-th closest vertex, so the $k$-subtree $T'$ of its shortest path trees in $G$ with $\phi(x,G) = \max_{v\in V(T')}d(x,v)$ is induced by its $k$ closest vertices. It thus follows that finding the local partial center $x^*_e$ on $e$ requires to compute the distance of every point on $e$ to its $k$-th closest vertex. As we shall show below, it is equivalent to solve the problem of computing the \textit{$k$-th level} of a set of $x$-monotone polygonal chains. 

Consider function $y = d(v,x)$ of each $v\in V$ for $x\in e$ in the $x,y$-coordinate plane. Function $y = d(v,x)$ defines a $x$-monotone polygon chain $C_v$ whose leftmost and rightmost endpoints are respectively on the vertical line $x=0$ and $x=l(e)$. Specifically, either $C_v$ is a line segment of slope $+1$ or $-1$, or it consists of two line segments where the left one is of slope $+1$ and the right one is of slope $-1$. We refer to the segment of slope $+1$ on $C_v$ as its \textit{$x$-segment}, and its segment of slope $-1$ as its \textit{$y$-segment}. (Because rotating the $x,y$-plane by $45$ degree along the positive 
$x$-axis causes that the segment of slope $+1$ becomes horizontal and the segment of slope $-1$ becomes vertical.)

Let $C$ be the set of the $n$ polygonal chains of all vertices w.r.t. $e$. $C$ can be obtained in $O(n)$ time by the given distance matrix. If values $d(v,r)$ of all vertices are distinct and so are values $d(v,s)$, then any two chains in $C$ intersect at most once. If so, $x^*_e$ is of the same $x$-coordinate as the lowest point on the \textit{$k$-th level} of an arrangement of $C$ that is the closure of the set of all points that lie on chains of $C$ such that the open downward-directed vertical ray emanating from this closure intersects exactly $k$ chains of $C$; this $k$-th level of $C$ can be computed in $O(n\log^2n+nk)$ time~\cite{ref:EverettAn96}. 

Generally, chains of $C$ may overlap with each other. Specifically, the $x$-segments (resp., the $y$-segment) of two vertices with same values $d(v,r)$ (resp., $d(v,s)$) have the same left (resp., right) endpoint, so they overlap partially or fully from their common left (resp., right) endpoint. This implies that the $k$-th level of $C$ may not exist since it requires exactly $k$ chains not above its closure. Hence, we give a more generalized definition: The (general) $k$-th level of $C$ is the closure of a set of all points that lie on chains of $C$ such that 
the open downward-directed vertical ray emanating from this closure intersects fewest but at least $k$ chains in $C$. See Fig.~\ref{fig:klevel} for an example. 

\begin{figure}
    \centering
     \includegraphics[scale=0.5]{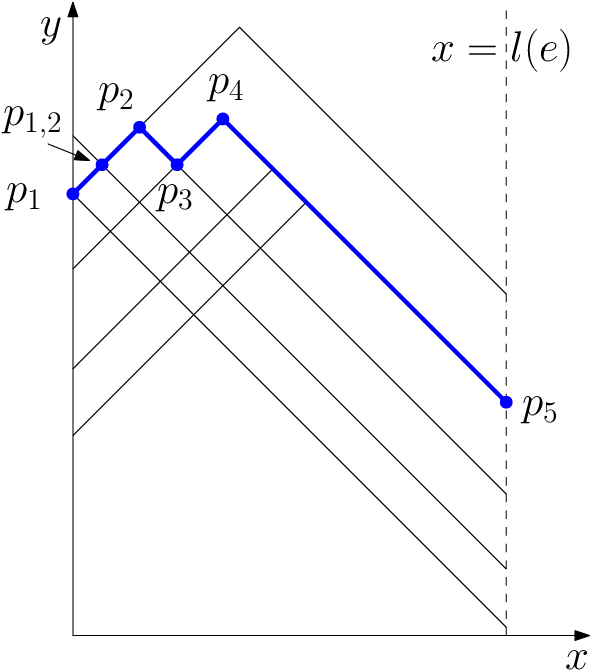}
     \caption{Illustrating the $6$-level closure (the heavy chain) of seven functions $D(v,x)$, which has five vertices $p_1, p_2, p_3, p_4, p_5$. It cannot turns at the intersection $p_{1,2}$.}
     \label{fig:klevel}
\end{figure}

Clearly, the (general) $k$-th level of $C$ always exists and $x^*_e$ is decided by its lowest point. Let $C^k$ denote the $k$-th level of $C$. We can utilize the following lemma to construct $C^k$, whose proof is in Section~\ref{sec:kthlevel}. 

\begin{lemma}\label{lem:klevelclosure}
    $C^k$ is a $x$-monotone polygon chain of complexity $O(n)$, and it can be constructed in $O(n\log n)$ time. 
\end{lemma}

Based on the above analysis, $x^*$ and $\lambda^*$ can be computed as follows: For each edge $e$ of $G$, we first determine in $O(n)$ time the set $C$ of functions $d(v,x)$ of each $v\in V$, and then apply Lemma~\ref{lem:klevelclosure} to find the lowest point on $C^k$ in $O(n\log n)$ time. Among all obtained points, we keep the one with the smallest $y$-coordinate since $\lambda^*$ is its $y$-coordinate and $x^*$ is its projection to the corresponding edge. The total running time is $O(mn\log n)$.  

\begin{theorem}\label{the:unweightedgeneralresult}
    The unweighted connected $k$-vertex one-center problem can be solved in $O(mn\log n)$ time. 
\end{theorem}


\subsection{Computing the $k$-th Level}\label{sec:kthlevel}
    The $k$-th level $C^k$ of $C$ starts with a chain in $C$ (its leftmost segment) that intersects line $x=0$ at the lowest point so that at least $k$ chains in $C$ intersect $x=0$ below or at this point. Let $C^+$ (resp., $C^-$) be the set of all chains whose $x$-segments 
    (resp., $y$-segments) intersect $x=0$ at this point. All chains of $C^-$ overlap fully; the $x$-segments of chains in $C^+$ overlap partially or fully 
    starting from this point, which is their common left endpoint. 
    If $C^-$ and chains below this point at $x=0$ exceed $k-1$ chains, then $C^k$ starts at this point and follows chains in $C^-$, i.e., their common $y$-segment. Otherwise, $C^k$ follows the overlapping $x$-segment of chains in $C^+$. 

    As $x$ increases, $C^k$ always follows the chains that it lies on except when one of the following situations occurs:    
    
    \begin{enumerate}
        \item $C^k$ reaches an endpoint of the chains on $x=l(e)$ where $C^k$ terminates; 
        \item $C^k$ reaches an intersection of two chains in $C$ where their slopes are of opposite signs.  
    \end{enumerate}

    It is clear to see that $C^k$ is a $x$-monotone polygonal chain (zigzag) of at most $2n$ line segments. More specifically, $C^k$ overlaps with the $x$- and $y$-segments of chains in $C$ alternatively; for any line segment on $C^k$ of slope $+1$ (resp., $-1$), the $x$-intercept of the line containing it is smaller than that of any succeeding line segment of slope $+1$ (resp., $-1$). These properties of $C^k$ support the correctness of the following algorithm for constructing $C^k$. 

    \paragraph{\textbf{Preprocessing}}
    To construct $C^k$, we compute two ordered sets (arrays) $S^+$ and $S^-$ for $C$. $S^+$ stores all $x$-segments of chains in $C$ in the following manner: Every element of $S^+$ is a sequence of overlapping $x$-segments sorted descendingly by the $y$-coordinates of their right endpoints; all sequences (arrays) of $S^+$ are in the ascending order by the $x$-intercepts of the lines containing their own $x$-segments, which is the left-to-right order of $C^k$ intersecting these (parallel) lines. Similarly, $S^-$ stores sequences of $y$-segments so that $y$-segments on the same line are in the same sequence where they are sorted descendingly by the $y$-coordinates of their left endpoints, and all sequences in $S^-$ are sorted ascendingly by the $x$-intercepts of the lines containing their $y$-segments, i.e., the left-to-right order of $C^k$ intersecting their (parallel) lines. Obviously, $S^-$ and $S^+$ can be derived in $O(n\log n)$ time. 

    \paragraph{\textbf{The algorithm}}
    Let $N$ be the complexity of $C^k$ and $l_i$ denote the $i$-th segment on $C^k$ in the left-to-right order. There is a (unique) sequence in $S^-\cup S^+$ where $l_i$ overlaps with its segments; for convenience, denote this sequence and its index in array $S^-$ or $S^+$ both by $I_i$. 
    Let $p_i$ be $l_i$'s left endpoint, which is the $i$-th vertex of $C^k$. For a line segment, we refer to the $x$-intercept of the line containing it as its $x$-intercept. For any two sequences respectively in $S^-$ and $S^+$, every pair of segments from different sequences intersect at most once and intersections of all pairs are at the same point. 
    
    By the properties of $C^k$, all segments $l_i$ of $C^k$ with odd indices have their sequences $I_i$ all in the same set of $S^+$ and $S^-$, and sequences $I_j$ of segments $l_j$ with even indices are all in the other set. Moreover, for each $1<i<N$, 
    only segments that are parallel to $l_{i-1}$ but of larger $x$-intercepts are likely to intersect $l_i$ and thereby overlap $l_{i+1}$. Hence, we have index $I_{2i-1}<I_{2i+1}$ for each $1\leq i\leq \frac{N-1}{2}$, 
    and $I_{2j}<I_{2j+2}$ for each $1\leq j\leq \frac{N-2}{2}$. 

    It follows that extending lines of all $x$- and $y$-segments forms a grid of the $x,y$-plane so that $C^k$ is a $x$-monotone path on the grid that starts from the intersection $p_1$ where segments in $I_1$ intersect $x=0$, follows in order the longest segment in each $I_i$ all the way to the right, turns upward or downward at the intersection between segments in $I_i$ and $I_{i+1}$ for each $1\leq i< N$, and ends at the intersection $p_{N+1}$ where segments in $I_N$ intersect $x=l(e)$. 
    
    Define function $f(l_i,x)$ as the number of chains in $C$ not above the line of segment $l_i$ of $C^k$ for $x\in [x(p_i), l(e)]$. $f(l_i,x)$ is a piecewise constant function and changes only at $x$-coordinates of intersections between $l_i$'s line and segments in the succeeding sequences of $I_{i-1}$ in the set containing $I_{i-1}$. If $l_i$ is of slope $-1$, $f(l_i,x)$ is monotonically decreasing as $x$ increases since $l_i$ intersects only $x$-segments; if $l_i$ is of slope $+1$, similarly, $f(l_i,x)$ must be monotonically increasing. Hence, $C^k$ must turn at a breakpoint of $f(l_i,x)$ if it turns to be less than $k$ after this breakpoint, or the turn of $C^k$ at this breakpoint reduces chains not above $C^k$ but still ensures $k$ chains below or on $C^k$.   
    
    Guided by the above analysis, we can loop through $S^+$ and $S^-$ alternatively to find all sequences $I_i$ and vertices $p_i$ of $C^k$ in order by computing the breakpoints of their functions $f(l_i,x)$. 


    Suppose that we are about to determine $I_{i+1}$ and $p_{i+1}$. While visiting sequences from entry $I_{i-1}+1$ in $S^+$ or $S^-$, we compute $f(l_i,x)$ at the intersection generated by every sequence and $I_i$ in order until all sequences after $I_{i-1}$ have been processed, or a sequence $I'$ is encountered 
    so that either of the following conditions satisfies: (1) If $I_i\in S^-$, $f(l_i,x)$ becomes less than $k$ after the intersection caused by $I_i$ and sequence $I'\in S^+$; 
    (2) If $I_i\in S^+$, the longest segment of $I'$ has at least $k$ chains not above it at any point of $x$-coordinate slightly larger than that of the intersection by $I'$ and $I_i$, i.e., turning $C^k$ at the intersection by $I'$ and $I_i$ still ensures at least $k$ chains below or on $C^k$. In the former situation, $l_i$ is exactly the last (rightmost) segment on $C^k$, so the last vertex is the intersection of the segments in $I_i$ and line $x=l(e)$. 
    In the later situation, $C^k$ must turn at the intersection by $I_i$ and $I'$, so $I_{i+1} = I'$ and $p_{i+1}$ is exactly this intersection. The details of handling the two cases in the later situation are presented as follows.  
    \begin{enumerate}
    \item Case $I_i\in S^-$.
    Sequence $I_{i+1}$ is after $I_{i-1}$ in $S^+$. Starting from index $I_{i-1}+1$, 
    we visit each sequence $I\in S^+$ to determine value $f(l_i,x)$ at the $x$-coordinate of the intersection by $I_i$ and $I$ until $I_{i+1}$ is found. Let $Q_i=\{p_{i,1} = p_i, p_{i,2}, \cdots\}$ be the ordered set of all intersections generated by $I_i$ and all sequences in $S^+$ from $I_{i-1}$ to the end. Note that if $l_i$ is not the last segment of $C^k$ then there is at least one sequence after $I_{i-1}$ in $S^+$ such that its longest segment intersects that of $I_i$'s. Denote by $I^i_j$ the sequence in $S^+$ whose longest $x$-segment intersects that of $I_i$ at $p_{i,j}$.   
        
    As illustrated in Fig.~\ref{fig:negativei}, if the longest segment of the current sequence $I$ does not intersect that of $I_i$, then $f(l_i,x)$ remains constant (since the whole chain of every $x$-segment 
    in $I$ remains below the line of $l_i$). Otherwise, $f(l_i,x)$ breaks at their intersection. It takes $O(1)$ time to decide if $I$'s longest segment intersects that of $I_i$ since $x$-segments in any sequence of $S^+$ are in the descending order of the $y$-coordinates of their right endpoints.  
        
        \begin{figure}[h]
        \centering
        \includegraphics[scale = 0.5]{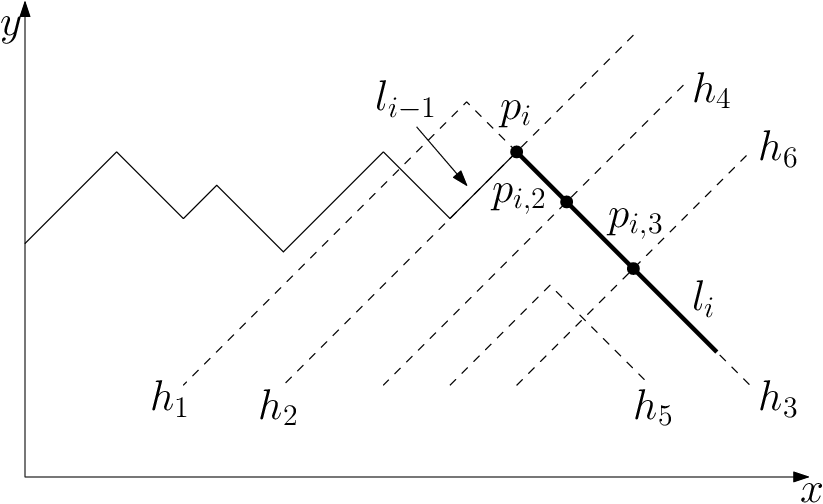}
        \caption{Illustrating the case of segment $l_i$ on $C^k$ is of slope $-1$: As $x$ increases from $x(p_{i,2})$, the $x$-segment of chain $h_4$ has its right endpoint higher than $p_{i,2}$ and so $h_4$ becomes above $l_i$ while chain $h_3$ is still on $l_i$ since the right endpoint of its $x$-segment is lower than or at $p_{i,2}$; moreover, the $x$-segment of chain $h_5$ does not intersect $l_i$ and so $h_5$ is below $C^k$, and the $x$-segment of chain $h_6$ is on $l_i$ at $x = x(p_{i,3})$; it follows that $f(l_i,x(p_{i,3})) = f(l_i,x(p_{i,2})) - 1$.}
        \label{fig:negativei}
        \end{figure}

    Generally, $I_i$ and $I$ has an intersection, and suppose their intersection is $p_{i,j+1}$. So, $I$ is sequence $I^i_{j+1}$ in $S^+$. At $x=x(p_{i,j})$, $f(l_i,x(p_{i,j}))$ equals to the total number of chains below $l_i$, chains whose only $y$-segments contain $p_{i,j}$, chains whose only $x$-segments contain $p_{i,j}$, and chains whose $x$-segments and $y$-segments both contain $p_{i,j}$ (i.e., whose peaks are $p_{i,j}$). For $x>x(p_{i,j})$, only chains whose only $x$-segments contain $p_{i,j}$ turn to be above the line of $l_i$; the $x$-segments of all these chains are all $x$-segments stored in sequence $I^i_{j}$ whose right endpoints are of $y$-coordinates larger than $y(p_{i,j})$ (i.e., higher than $p_{i,j}$). In addition, every chain with $x$-segments in sequence $I = I^i_{j+1}$ is below $l_i$'s line at $x=x(p_{i,j})$. Thus, $f(l_i,x)$ at any $x\in (x(p_{i,j}), x(p_{i,j+1})]$ equals to $f(l_i,x(p_{i,j}))$ minus the number of $x$-segments in $I^i_j$ whose right endpoints are higher than $p_{i,j}$. It implies that $f(l_i,x(p_{i,j+1}))$ can be obtained in $O(|I^i_j|)$ time by scanning $I^i_j$ to count all such $x$-segments. 
        
    

    Now we are ready to present the routine for computing every breakpoint of $f(l_i,x)$. Through the whole procedure, we always maintain the last breakpoint $f'$ obtained, its corresponding sequence (index) $I'$ in $S^+$, and the intersection $p'$ by $I'$ and $l_i$. At the beginning, $p'= p_i = p_{i,1}$ and $f' = f(l_{i-1},x(p_{i,1}))$, which have been set at the end of the iteration for determining $p_i$ and $I_i$. So, we only set $I' = I_{i-1}$ initially. 
        
    Starting from index $I_{i-1}+1$, for each sequence $I$ in $S^+$, we first check in $O(1)$ time if the longest segments of $I_i$ and $I$ intersect. If no, we skip it and continue to process next sequence in order. Otherwise, we proceed to determine in $O(|I'|)$ time $f(l_i,x)$ at their intersection $q$ as follows. Scan sequence $I'$ to count the $x$-segments with higher right endpoints than $p'$, so $f(l_i,x)$ at any $x\in (x(p'), x(q)]$ equals to $f'$ minus the total number of such $x$-segments. If $f(l_i,x(q))\geq k$, then $p'$ is not $p_{i+1}$ and $C^k$ turns at most at $q$; hence, after setting $I' = I$, $p' = q$, and $f' = f(l_i,x(q))$, we continue to process next sequence of $S^+$. Otherwise, $f(l_i,x(q))<k$, which means that $C^k$ must turn at the last intersection $p'$. So, we set $p_{i+1}=p'$ and $I_{i+1} =I'$. 
    
    It should be mentioned that at the end of the iteration for determining $p_{i+1}$ and $I_{i+1}$, variable $p'$ and $f'$, which are respectively $p_{i+1}$ and $f(l_{i+1}, x(p_{i+1}))$, have been initialized properly for determining $f(l_{i+1},x)$ further.
            
    Regarding to the time complexity, it should be noted that sequences between $I_{i+1}+1$ and the last visited $I''$ in $S^+$ that causes $f(l_i,x)<k$ will be visited again as determining the breakpoints of $f(l_{i+2},x)$ to find $p_{i+3}$ and $I_{i+3}$ further. Every sequence between $I_{i+1}+1$ and $I''-1$ in $S^+$ has its $x$-segment not intersect $l_i$'s line and their chains are all below $C^k$. But the longest segment of $I''$ may intersect $l_{i+2}$. If so, $I''$ must be the first sequence after $I_{i+1}$ in $S^+$ whose $x$-segment intersects $l_{i+2}$, which means index $I''\leq I_{i+3}$ in $S^+$. Otherwise, they do not intersect and $I''< I_{i+3}$. Due to $I''\leq I_{i+3}$, after finding $I_{i+3}$, it is not likely to visit again these sequences between $I_{i+1}+1$ and $I''$ in $S^+$. Additionally, in both iterations, $I''$ is processed by counting $x$-segments in $I_{i+1}$ of right endpoints respectively higher than $p_{i+1}$ and $p_{i+2}$. Thus, we can charge the time on processing every sequence between $I_{i+1}+1$ and $I''$ in $S^+$ for finding $I_{i+1}$ to the iteration of finding $p_{i+3}$ and $I_{i+3}$.  
    
    As a consequence, $p_{i+1}$ and $I_{i+1}$ can be found in the time linear to the total size of sequences between $I_{i-1}$ and $I_{i+1}$ in $S^+$. 

    \item Case $I_i\in S^+$. Sequence $I_{i+1}$ is after entry $I_{i-1}$ in $S^-$. $f(l_i,x)$ breaks only at the $x$-coordinates of intersections caused by $I_i$ and sequences in $S^-$ from $I_{i-1}$ to the end. Recall that $Q_i = \{p_{i,1} = p_i, p_{i,2}, \cdots\}$ is the ordered set of all these intersections. For each intersection $p_{i,j}$, $f(l_i,x(p_{i,j}))$ equals to the total number of chains below $l_i$, chains whose only $x$-segments contain $p_{i,j}$, chains whose only $y$-segments contain $p_{i,j}$, and chains whose peaks are at $p_{i,j}$.

        \begin{figure}[h]
        \centering
        \includegraphics[scale = 0.5]{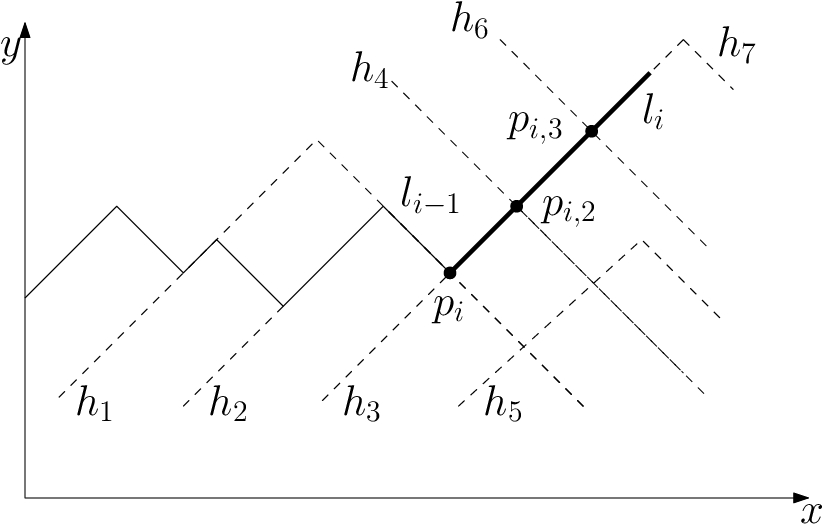}
        \caption{Illustrating the case of segment $l_i$ on $C^k$ is of slope $+1$: As $x$ increases from $x(p_{i,2})$, chains $h_1$ and $h_2$ are still below $l_i$ and the $y$-segments of chain $h_3$ and $h_4$ intersecting $l_i$ at $p_{i,2}$ become below $l_i$; moreover, chain $h_5$ whose $y$-segment does not intersect $l_i$ is below the line of $l_i$ but chain $h_6$ that is above the line of $l_i$ at any $x<x(p_{i,3})$ intersects $l_i$ at $p_{i,3}$ and become below $C^k$ at any $x>x(p_{i,3})$; 
        it follows that $f(l_i,x(p_{i,3})) = f(l_i,x(p_{i,2})) +1$.
        Further, if $C^k$ turns at $p_{i,3}$ then chain $h_7$ on $l_i$ would be above 
        $l_{i+1}$ and so $f(l_{i+1},x(p_{i+1,2}))=f(l_i,x(p_{i,3})) -1$.}
        \label{fig:positivei}
        \end{figure}
    
    As illustrated in Fig.~\ref{fig:positivei}, at any $x\in (x(p_{i,j}), x(p_{i,j+1})]$, chains that are below $l_i$ at $x=x(p_{i,j})$ are below $l_i$'s line; so are chains whose peaks are at $p_{i,j}$ and chains whose only $y$-segments contain $p_{i,j}$; chains whose only $x$-segments contain $p_{i,j}$ are on $l_i$'s line at any $x\in (x(p_{i,j}), x(p_{i,j+1})]$. So, $f(l_i,x) = f(l_i,x(p_{i,j}))$ at any $x\in [x(p_{i,j}), x(p_{i,j+1}))$. However, in sequence $I^i_{j+1}$ that generates $p_{i,j+1}$ with $I_i$, every $y$-segment that intersects $l_i$'s line (at $p_{i,j+1}$) and is of a higher left endpoint than $p_{i,j+1}$ turns to be below $l_i$'s line for $x>x(p_{i,j+1})$ while its chain is above $l_i$'s line at any $x\in [x(p_{i,1}), x(p_{i,j+1}))$. Hence, $f(l_i,x(p_{i,j+1}))$ equals to $f(l_i,x(p_{i,j}))$ plus the number of all such $y$-segments in $I^i_{j+1}$. 

    In addition, $C^k$ must turn at $p_{i,j}$ if the line containing $I_j^i$'s $y$-segments has at least $k$ chains below or on it at any $x\in (x(p_{i,j}), x')$ where $x'$ is slightly larger than $x(p_{i,j})$. Since that line is of slope $-1$, the number of chains below or on it at any $x\in (x(p_{i,j}), x')$ is equal to $f(l_i, x(p_{i,j}))$ minus the number of $x$-segments in $I_i$ whose right endpoints are higher than $p_{i,j}$.          

    We now present how to determine $f(l_i,x)$ and find $p_{i+1}$ and $I_{i+1}$. Recall that variables $f'$, $I'$ and $p'$ are used to maintain the last obtained breakpoint of $f(l_i,x)$, the sequence in $S^-$ that decides this breakpoint, and the intersection between the longest segments of $I_i$ and $I'$, respectively. $p'$, $I'$ and $f'$ were set respectively as $p_i$, $f(l_i,x(p_i))$, and $I_i$ at the end of the iteration of computing $p_i$ and $I_i$. So, we only set $I' = I_{i-1}$ initially. 

    Starting from index $I_{i-1}+1$, for every sequence $I$ in $S^-$, we first decide whether $I$'s longest $y$-segment intersects $I_i$'s. If no, we continue to visit next sequence in $S^-$. Otherwise, supposing their intersection is $p_{i,j+1}$, we proceed to compute $f(l_i,x(p_{i,j+1}))$ as follows: Scan sequence $I$, i.e., $I^i_{j+1}$, to count the total number of $y$-segments whose left endpoints are higher than $p_{i,j+1}$; next, set $f(l_i,x(p_{i,j+1}))$ as $f'$ plus the obtained number. 
        
    Proceed to determine whether $C^k$ turns at $p_{i,j+1}$. First, we compute the total number of $x$-segments in $I_i$ whose right endpoints are of larger $y$-coordinates than $p_{i,j+1}$. As we shall show later, this can be achieved in the time linear to the number of $x$-segments in $I_i$ whose right endpoints have $y$-coordinates falling in $(y(p_{i,j}), y(p_{i,j+1})]$. Then, we decide if $f'$ minus the obtained number is less than $k$. If yes, $p_{i,j+1}$ is not $p_{i+1}$, so we set $p' = p_{i,j+1}$, $f' = f(l_i,x(p_{i,j+1}))$, and $I'=I$, and continue to visit next sequence in $S^-$. Otherwise, $p_{i+1}=p_{i,j+1}$ and $I_{i+1} = I$, so we terminate this iteration after setting $p' = p_{i,j+1}$, $f' = f(l_i,x(p_{i,j+1}))$ and $I'=I$. 

    To decide if $p_{i+1}$ is $p_{i,j+1}$, we need to compute the total number of $x$-segments in $I_i$ whose right endpoints are higher than $p_{i,j+1}$. Recall that all $x$-segments in $I_i$ are sorted descendingly by the $y$-coordinates of their right endpoints. To achieve the running time specified above, we maintain an additional variable $j'$ through the whole procedure such that $j'$ is the index of the last $x$-segment in $I_i$ whose right endpoint is higher than $p'$. (Every $x$-segment in $I_i$ after position $j'$ is of a right endpoint lower than or same as $p'$.) At the beginning, besides of setting $I'$, we also scan sequence $I_i$ to find the last $x$-segment whose right endpoint is higher than $p'=p_i$, and subsequently, initialize $j'$ as the index of that segment. 

    While finding $p_{i+1}$, for every intersection $p_{i,j+1}$, since $y(p_{i,j+1})$ is no less than the $y$-coordinate of the right endpoint of the last $x$-segment in $I_i$ that is higher than $p'=p_{i,j}$, which is the $x$-segment at index $j'$ in $I_i$, we loop $I_i$ backward from index $j'+1$ to the beginning to find the first $x$-segment with a higher right endpoint than $p_{i,j+1}$, which is the last $x$-segment in $I_i$ whose right endpoint is higher than $p_{i,j+1}$. Hence, the total number of all $x$-segments in $I_i$ with right endpoints higher than $p_{i,j+1}$ can known in the time linear to the number of $x$-segments in $I_i$ whose right endpoints are of $y$-coordinates in $(y(p_{i,j}), y(p_{i,j+1})]$. As a result, the total time for counting such $x$-segments in $I_i$ for all intersections $p_{i,1}, p_{i,2}, \cdots, p_{i,t}=p_{i+1}$ is linear to the size of $I_i$. 
    
    Overall, when $l_i$ is of slope $+1$, the total running time is linear to the total size of sequences between $I_{i-1}$ and $I_{i+1}$ in $S^-$ and $I_i$.
    \end{enumerate}

    \paragraph{\textbf{Wrapping things up}} Start by computing $p_1$, $I_1$, and $f(l_1,x(p_1))$ as follows. 
    Recall that $p_1$ is the lowest intersection between chains in $C$ and line $x=0$ where at least $k$ chains 
    are not above it at $x=0$. $p_1$ and $f(l_1,x(p_1))$ can be obtained in $O(n)$ time as follows: 
    Compute the $y$-coordinates of intersections between line segments in $S^-\cup S^+$ 
    and line $x=0$, find their $k$-th smallest value by the Selection algorithm~\cite{ref:CormenIn22}, 
    and then count how many $y$-coordinates are no more than it. 
    
    Next, we determine $I_1$ in $O(n)$ time. Recall that $C^-$ (resp., $C^+$) is the subset of 
    chains in $C$ whose $y$-segments (resp., $x$-segments) intersect $x=0$ at $p_1$. Clearly, 
    the $y$-segments of chains in $C^-$ belong to a sequence $I^-$ of $S^-$ such that lines of $y$-segments in every preceding sequence of $I^-$ 
    are below $p_1$ but lines of $y$-segments in every sequence after $I^-$ are above $p_1$; 
    whereas the $x$-segments of $C^+$ are all in a sequence $I^+$ of $S^+$ so that lines of $x$-segments in its each preceding sequence are above $p_1$ but lines of $x$-segments in its each succeeding sequence are below $p_1$. In order to find $I_1$, 
    we scan $S^-$ and $S^+$ again from the beginning to find $I^-$ and $I^+$ by checking whether line $x=0$ 
    intersect the longest segment of each sequence at $p_1$. 
    Simultaneously, we count all $x$- and $y$-segments that intersect line $x=0$ below $p_1$, 
    and count all $y$-segments that intersect $x=0$ exactly at $p_1$. 
    Let $n'$ and $n''$ be the two obtained numbers, respectively. 
    Because every chain containing both $x$-segment and $y$-segment has its peak 
    on neither line $x=0$ nor line $x=l(e)$. Thus, if $n'+n''\geq k$ then $I_1$ is $I^-$, 
    and otherwise, $I_1$ is $I^+$. 

    Initialize $p' = p_1$ and $f' = f(l_1,x(p_1))$. Next, if $I_1 = I^-$, we set $I' = I^+$ but set $I'$ as the first sequence of $S^+$ when $I^+$ is null, that is, no $x$-segments intersect $x=0$ at $p_1$; otherwise, $I_1 = I^+$, so we set $I' = I^-$ if $I^-$ is not null and otherwise, set $I'$ as the first sequence of $S^-$.

    Proceed to loop $S^+$ and $S^-$ alternatively to find $p_i$ and $I_i$ for each $1<i\leq N$ in order. 
    Depending on whether $I_1$ is in $S^+$ or $S^-$, we perform the two above routines alternatively to find vertices on $C^k$ from left to right. Note that once either $S^+$ or $S^-$ runs out during the procedure, we immediately compute the intersection between the longest segment of the current sequence and line $x=l(e)$, which is exactly the last vertex of $C^k$, and then terminate the algorithm. Based on the above analysis, the total running time is linear to the number of segments in $S^+\cup S^-$, which is $O(n)$. 
    
    Recall that we spend $O(n\log n)$ time on computing $S^-$ and $S^+$ in the preprocessing work. In a sum, the $k$-th level $C^k$ of $C$ can be computed in $O(n\log n)$ time. Thus, Lemma~\ref{lem:klevelclosure} is proved. 

\section{Solving the Problem on a Tree}\label{sec:treeversion}
In this section, we propose two faster algorithms for the problem on tree graphs respectively in the weighted and unweighted case, where the algorithm for the unweighted case is presented in the proof of Theorem~\ref{the:unweightedtree}. Note that the intervertex distance matrix of the tree graph is not given. 

Let $T$ represent the given tree graph, and $R(T)$ be its root. So, $|T|=O(n)$. In the preprocessing, we compute and maintain the distance of every vertex to $R(T)$, and then apply the lowest common ancestor data structure~\cite{ref:BenderTh04} to $T$ so that given any two points $y,y'$ on $T$, $d(y,y')$ can be known in $O(1)$ time. The preprocessing time is $O(n)$. 

Below, we shall first give in Section~\ref{sec:treefastertest} a faster feasibility test for tree graphs, and then present in Section~\ref{sec:treecomputelambda} an algorithm based on tree decomposition techniques for computing $\lambda^*$. 

\subsection{A Faster Feasibility Test on Trees}\label{sec:treefastertest}
Recall that on a general graph, to decide the feasibility of any given $\lambda$, Lemma~\ref{lem:generalweightedfeasibility} is applied to every edge to find a point whose largest self-inclusive subtree covered by it (under $\lambda$) is of size at least $k$. On tree graphs, however, the following shows that only $O(n)$ points on $T$ need to be considered. 

For each vertex $v$, $D(v,x)$ increases as $x$ moves away from $v$ along any path. 
There is a point $x'$ on path $\pi(v,R(T))$ so that $D(v,x')=\lambda$ if $D(v,R(T))\leq\lambda$ and otherwise, $D(v,x)$ achieves its maximum at $x'$, i.e., $x'=R(T)$. We refer to $x'$ as the \textit{critical} point of $v$ (w.r.t. $R(T)$). Denote by $Q_\lambda$ the set of all $n$ critical points. We have the following observation for $Q_\lambda$. 

\begin{observation}\label{obs:treecriticalpoints}
    There must be a point in $Q_\lambda$ so that the largest self-inclusive subtree covered by it is of size at least $k$ if and only if $\lambda$ is feasible. 
\end{observation}
\begin{proof}
    We show below that if $\lambda$ is feasible then $Q_\lambda$ must contains such a point. Due to $\lambda\geq\lambda^*$, by Observation~\ref{obs:decision}, there exists a $k$-subtree $T'$ of $T$ so that $T'$ can be covered by a point on $T'$ under $\lambda$. 

    Suppose $T'$ is rooted at vertex $r'$. Consider the $p$-center problem for $T'$ w.r.t. $V(T')$ with $1\leq p\leq k$. Under any given value $\lambda'$, the feasibility test~\cite{ref:WangAn21} for the $p$-center problem traverses $T'$ from $r'$ in the post order to place least centers in a greedy way so that every center is placed at a critical point of $V(T')$ w.r.t. $r'$. When $p=1$, since $\lambda$ is feasible, the center placed by~\cite{ref:WangAn21} under $\lambda$ covers $T'$. 

    If $r'=R(T)$ or $\max_{v\in V(T')}w_vd(v,r')\geq\lambda$, then the placed center in $T'$ must be in $Q_\lambda$. Otherwise, $r'\neq R(T)$ and $\max_{v\in V(T')}w_vd(v,r')<\lambda$, which means the critical point of each $v\in V(T')$ w.r.t. $R(T)$ is in $\pi(r',R(T))/\{r'\}$. 
    Let $c'$ be the lowest critical point of vertices in $V(T')\cup V(\pi(r',R(T)))$, and $c'$ is on $\pi(r',R(T))/\{r'\}$. Clearly, $c'\in Q_\lambda$ and it can cover all vertices in $V(T'\cup \pi(r',c'))$, which induce a connected subtree of size at least $k$ that is adjacent to $c'$ on $T$. Thus, the observation holds. \qed
\end{proof}
 
Our algorithm decides the feasibility of $\lambda$ as follows: Traverse $T$ to compute set $Q_\lambda$. Next, for each $x\in Q_\lambda$, we determine the size of the largest $x$-inclusive subtree $T_\lambda(x)$ covered by $x$. If a point in $Q_\lambda$ has $|T_\lambda(x)|\geq k$, then $\lambda$ is feasible. Otherwise, it is infeasible. 

To compute $Q_\lambda$, we traverse $T$ from $R(T)$ in the post-order, and during the traversal, we maintain the path from $R(T)$ to the current vertex by employing a stack. For each vertex encountered, we perform a binary search on its path to $R(T)$ to compute its critical point. Since the distance of any two points on $T$ can be obtained in $O(1)$ time, the critical point of every vertex can be figured out in $O(\log n)$ time. Hence, computing $Q_\lambda$ takes $O(n\log n)$ time. 

It remains to solve the query problems of counting and reporting $V(T_\lambda(x))$ for any given point $x\in T$. If $T$ is a balanced binary tree, which can be verified in $O(n)$ time, Lemma~\ref{lem:Bquerylargestsubtree} can be employed to construct a data structure $\calA_1$ in $O(n\log n)$ time that answers the two queries in $O(\log^2n)$ and $O(\log^2n + |V(T_\lambda(x))|)$ time, respectively. In general, $T$ is a general tree graph. Then, Lemma~\ref{lem:querylargestsubtree} can be applied to build a data structure $\calA_2$ in $O(n\log n)$ time that counts and reports $V(T_\lambda(x))$ for any $x\in T$ respectively in $O(\log^2n)$ and $O(\log^2n + |V(T_\lambda(x))|)$ time. Note that compared to $\calA_2$, the construction of $\calA_1$ and the query on it are much simpler. 

\begin{lemma}\label{lem:Bquerylargestsubtree}
    For $T$ being a balanced binary tree, we can build a data structure $\calA_1$ in $O(n\log n)$ time that answers the counting query of $V(T_\lambda(x))$ in $O(\log^2 n)$ time, and reports $V(T_\lambda(x))$ in $O(\log^2 n + K)$ time, where $K=|V(T_\lambda(x))|$. 
\end{lemma}

\begin{lemma}\label{lem:querylargestsubtree}
    For $T$ being a general tree, we can build a data structure $\calA_2$ 
    in $O(n\log n)$ time that answers the counting query on $V(T_\lambda(x))$ 
    in $O(\log^2 n)$ time, and reports $V(T_\lambda(x))$ in $O(\log^2 n + K)$ time. 
\end{lemma}

Regarding to our problem, the goal is to decide for any given point $x\in T$ whether $|V(T_\lambda(x))|\geq k$. By maintaining partial but enough information on $\calA_1$ and $\calA_2$, the counting query on $V(T_\lambda(x))$ can be answered in $O(\log n\log k)$ time. This improvement is demonstrated in Corollary~\ref{cor:fasterkquery}. Note that Corollary~\ref{cor:fasterkquery} does not support reporting $V(T_\lambda(x))$. 

\begin{corollary}\label{cor:fasterkquery}
For any given point $x$ on $T$, with $O(n\log n)$-time preprocessing work, we can decide in $O(\log n\log k)$ whether $T_\lambda(x)$ is of size at least $k$. 
\end{corollary}

The proofs of Lemma~\ref{lem:Bquerylargestsubtree}, Lemma~\ref{lem:querylargestsubtree}, and Corollary~\ref{cor:fasterkquery} are presented respectively in Section~\ref{sec:buildingA1}, Section~\ref{sec:buildingA2}, and Section~\ref{sec:improveA1A2}. By Corollary~\ref{cor:fasterkquery}, we have the following lemma for the feasibility test. 

\begin{lemma}\label{lem:treedecision}
    The feasibility test on trees can be solved in $O(n\log n\log k)$ time.
\end{lemma}

\subsubsection{The Data Structure $\calA_1$}\label{sec:buildingA1}

For any point $x$ on an arbitrary edge $e(r,s)$ of $T$ where $r$ is the parent of $s$, on path $\pi(s,R(T))$, there is a longest (upward) subpath $\pi_x=\{x, v_1, \cdots, v_z\}$ covered by $x$ such that the parent of $v_z$ is the lowest vertex on $\pi(x, R(T))$ whose weighted distance to $x$ is larger than $\lambda$, that is, the lowest heavy vertex of $x$ on $\pi(x, R(T))$. For each $2\leq i\leq z$, a subtree hangs off $\pi_x$ at $v_i$ and it is the (whole) subtree of $T$ rooted at $v_i$'s child $u_i$ not on $\pi_x$ (while the other child of $v_i$ is $v_{i-1}$ on $\pi_{x}$). $v_1$ has both children not on $\pi_x/\{x\}$; let $u_0$ be its child adjacent to $x$, i.e., $s$, and let $u_1$ be the other. Let $T_v$ represent the (whole) subtree of $T$ rooted at a vertex $v$. As illustrated in Fig.~\ref{fig:binarytreequery}, $T_\lambda(x)$ is the union of 
$\pi_x$ and the largest $u_i$-inclusive subtree in $T_{u_i}$ covered by $x$ for each $0\leq i\leq z$. 

    \begin{figure}[h]
        \centering
        \includegraphics[scale = 0.65]{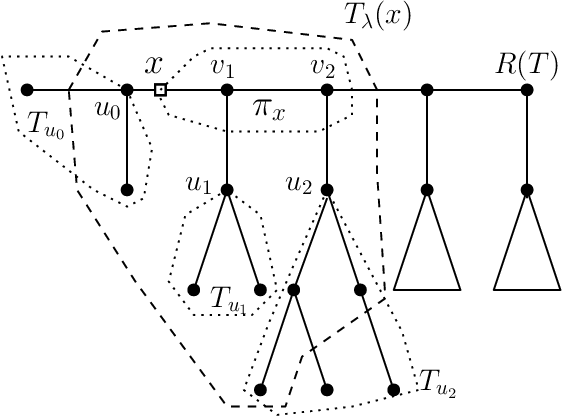}
        \caption{Illustrating the counting and reporting query of $T_\lambda(x)$ for $T$ being a balanced binary search tree.} \label{fig:binarytreequery}
    \end{figure}

    Clearly, $\pi_x$ can be found in $O(\log n)$ time during the traceback of searching for point $x$ on $T$. Simultaneously, for each $v_i$, we need to answer the counting and reporting queries about the largest $u_i$-inclusive subtree in $T_{u_i}$ covered by a given point $x\in \{u_i\}\cup T/T_{u_1}$. To this end, we construct the following functions for each $v\in V$ to support the two queries about the largest $v$-inclusive subtree on $T_v$. 
   
    Define $f(v,x)$ as the size of the largest $v$-inclusive (root-inclusive) subtree in $T_v$ covered by any point $x$ out of $T_v/\{v\}$ (any `outside' point). Clearly, $f(v,x)$ is a piecewise constant function w.r.t. distance $d(x,v)$; 
    as $x$ moves away from $v$, $f(v,x)$ monotonically decreases. For convenience, we use $x$ to denote $d(x,v)$ when we talk about $f(v,x)$. 

    In the preprocessing, we construct for each $v\in V$ function $f(v,x)$ by computing the ordered set of all its breakpoints so that the counting query on the largest $v$-inclusive subtree covered by any point out of $T_v/\{v\}$ can be answered in $O(\log n)$ time. 
    In addition, to support a linear-time reporting, for every breakpoint of $f(v,x)$, 
    we find the set $Q(v,x)$ that includes all additional vertices in the largest 
    covered $v$-inclusive subtree compared with the previous breakpoint 
    of a larger $x$-coordinate, i.e., $d(v,x)$. 

    We now present how to construct $f(v,x)$ and find sets $Q(v,x)$ for each $v\in V$ during the post-order on $T$. We use an array $F_v$ to store all breakpoints and their sets $Q(v,x)$ in the descending order by $x$-coordinates of breakpoints. 

    \begin{enumerate}
    
    \item The leaf case. Clearly, $f(v,x) = 1$ for $x\in [0,\frac{\lambda}{w_v}]$ and $f(v,x) = 0$ for $x\in (\frac{\lambda}{w_v},+\infty]$. To support a binary search on $f(v,x)$, we create array $F_v$ of size $3$ and store in order three tuples (breakpoints) $(x=+\infty, y=0, Q=\emptyset)$, $(x=\frac{\lambda}{w_v}, y=1, Q = \{v\})$, and $(x=0, y=1, Q=\emptyset)$, where variables $x$ store the $x$-coordinates of the breakpoints, variables $y$ are their values $f(v,x)$, and sets $Q$ are their sets $Q(v,x)$. Hence, for any leaf of $T$, $f(v,x)$ and $Q(v,x)$ can be constructed in $O(1)$ time. 
    
    \item The internal-node case. In general, $v$ has the two children $u_1$ and $u_2$. The breakpoints of function $f(u_1,x)$ (resp., $f(u_2,x)$) and their sets $Q(u_1,x)$ (resp., $Q(u_2,x)$) have been computed and are stored in array $F_{u_1}$ (resp., $F_{u_2}$). Clearly, for subtree $T_v$, $f(v,x) = 0$ at any $x\in (\frac{\lambda}{w_v}, +\infty)$ while at any $x\in [0, \frac{\lambda}{w_v}]$, $f(v,x) = 1 + f(u_1, x + d(u_1,v)) + f(u_2, x + d(u_2,v))$. 
    In addition, $f(v,x)$ breaks at $x = \frac{\lambda}{w_v}$ and all breakpoints of $f(u_1,x)$ and $f(u_2,x)$ whose $x$-values (in their tuples) are respectively in $[d(u_1,v), d(u_1,v) + \frac{\lambda}{w_v}]$ and $[d(u_2,v), d(u_2,v) + \frac{\lambda}{w_v}]$. 
    Hence, we can merge $F_{u_1}$ and $F_{u_2}$ as the following to determine every breakpoint of $f(v,x)$ and $Q(v,x)$ in order. 
    
    First, we compute $f(v,\frac{\lambda}{w_v})$ by performing a binary search 
    respectively on $F_{u_1}$ and $F_{u_2}$ to compute $f(u_1,\frac{\lambda}{w_v}+ d(u_1,v))$ and 
    $f(u_2,\frac{\lambda}{w_v}+ d(u_2,v))$. Two indices $i'$ and $j'$ are obtained so that $F_{u_1}[i'-1].x\geq \frac{\lambda}{w_v}+ d(u_1,v) > F_{u_1}[i'].x$ and $F_{u_2}[j'-1].x \geq \frac{\lambda}{w_v}+ d(u_2,v)> F_{u_2}[j'].x$. Hence, $f(v,\frac{\lambda}{w_v})$ equals to $ (1+ F_{u_1}[i'-1].y + F_{u_2}[j'-1].y)$. Additionally, $Q(v,\frac{\lambda}{w_v})$ is the union of $ \{v\}$, $\cup_{l=1}^{i'-1}F_{u_1}[l].Q$, and $\cup_{l=1}^{j'-1}F_{u_2}[l].Q$, which can be obtained in $O(|T_{u_1}| + |T_{u_2}|)$ time. Then, we push into an auxiliary queue $S$ this tuple $(x = \frac{\lambda}{w_v}, y= 1 + F_{u_1}[i'-1].y + F_{u_2}[j'-1].y, Q=Q(v,\frac{\lambda}{w_v}))$. 
    
    Proceed with merging arrays $F_{u_1}$ and $F_{u_2}$ respectively 
    staring from $i'$ and $j'$ to determine all other breakpoints of $f(v,x)$ in order. 
    Every step we compare $F_{u_1}[i].x - d(u_1,v)$ and $F_{u_2}[j].x - d(u_2,v)$ 
    and insert a breakpoint (tuple) into $S$ for each of the following cases. 

    \begin{enumerate}
        \item Case $F_{u_1}[i].x - d(u_1,v) < F_{u_2}[j].x -d(u_2,v)$: The next breakpoint of $f(v,x)$ in order is of $x$-coordinate $x'=F_{u_2}[j].x - d(u_2,v)$. On account of $f(u_1, x'+d(u_1,v)) = F_{u_1}[i-1].y$, $f(v,x') = F_{u_1}[i-1].y + F_{u_2}[j].y +1$. Since the breakpoint $F_{u_1}[i-1]$ leads a breakpoint of $f(v,x)$ preceding this current one in $S$, $Q(v,x')$ is $Q(u_2, F_{u_2}[j].x)$, which is $F_{u_2}[j].Q$. Thus, we insert a tuple 
        $(x = x', y = F_{u_1}[i-1].y + F_{u_2}[j].y +1, Q=F_{u_2}[j].Q)$ 
        into $S$ and then increment $j$. 
        \item Case $F_{u_1}[i].x - d(u_1,v) > F_{u_2}[j].x -d(u_2,v)$: The $x$-coordinate $x'$ of next breakpoint equals to $F_{u_1}[i].x - d(u_1,v)$. Similarly, $f(v,x')$ equals to $F_{u_1}[i].y + F_{u_2}[j-1].y + 1$, and $Q(v, x') $ is exactly $F_{u_1}[i].Q$. 
        For this case, we insert a tuple $(x = F_{u_1}[i].x - d(u_1,v), y = F_{u_1}[i].y + F_{u_2}[j-1].y +1, Q= F_{u_1}[i].Q )$ into $S$ and then increment $i$. 
        \item Case $F_{u_1}[i].x - d(u_1,v) = F_{u_2}[j].x -d(u_2,v)$: We insert into $S$ a tuple $(x = F_{u_1}[i].x - d(u_1,v), y = F_{u_1}[i].y + F_{u_2}[j].y +1, Q = F_{u_1}[i].Q\cup F_{u_2}[j].Q)$. Last, increment both $i$ and $j$. 
    \end{enumerate}

    We stop the merging once either $F_{u_1}[i].x - d(u_1,v)< 0$ or $F_{u_2}[j].x - d(u_2,v) <0$. 
    In the former case, for each remaining tuple of $F_{u_2}$ 
    with $F_{u_2}[j].x - d(u_2,v) \geq 0$, we insert 
    into $S$ a breakpoint $(x= F_{u_2}[j].x - d(v, u_2), y = 1 + F_{u_2}[j].y +F_{u_1}[i-1].y, Q = F_{u_2}[j].Q)$. In the later case, for each remaining tuple in $F_{u_1}$ satisfying 
    $F_{u_1}[i].x - d(u_1,v) \geq 0$, we insert into $S$ a breakpoint $(x= F_{u_1}[i].x - d(v, u_1), y = 1 + F_{u_1}[i].y +F_{u_2}[j-1].y, Q = F_{u_1}[i].Q)$. Subsequently, if the last breakpoint in $S$ is of $x$-coordinate larger than zero, supposing its $y$-coordinate is $y'$, 
    we insert into $S$ a tuple $(x=0, y=y', Q=\emptyset)$. 
    
    Last, we create an array $F_{v}$ of size $|S|+1$. Store a tuple $(x=+\infty, y = 0, Q=\emptyset)$ in the first position, and then copy all tuples of $S$ to $F_v$ following the popping order. The case where $v$ has only one child can be handled in a similar manner, so we omit the details. 

    Recall that at the beginning, we have preprocessed $T$ in $O(n)$ time so that the distance query for any two points on $T$ can be answered in $O(1)$ time. Regarding to the time complexity, except for the $O(\log |T_v|)$ time on computing the first breakpoint of $f(v,x)$, every other breakpoint can obtained $O(1)$ time. 
    Since $f(v,x)$ has at most $|T_v|$ breakpoints, the time complexity for constructing $f(v,x)$ is $O(|T_v|)$ time; the time of determining set $Q$ for all breakpoints of $f(v,x)$ 
    is $O(|T_v|)$ in total. 
    \end{enumerate}

    
    Since $T$ is a balanced binary tree, based on the above analysis, it is easy to see that $\calA_1$ can be constructed in $O(n\log n)$ time. 
    
    For any point $x$ of $T$, the counting and reporting queries about $T_\lambda(x)$ can be answered by $\calA_1$ as follows. Start from $R(T)$, we search for $x$ on $T$ by always visiting the child with a smaller distance to $x$, which can be known in $O(1)$ time. Hence, it takes $O(\log n)$ time to find the edge on $T$ containing $x$. During the traceback, starting from the preceding vertex of $x$ on $\pi_x$, we perform the following routine for each vertex until a heavy vertex of $x$ is met (whose weighted distance to $x$ is larger than $\lambda$). 

    For each non-heavy vertex $v$ encountered, due to $v\in\pi_x$ and $D(v,x)\leq\lambda$, we increment the counter $C$ and report $v$; subsequently, for $v$'s each child $u$ not on $\pi_x/\{x\}$, we perform a binary search in $F_u$ to compute $f(u,d(x,u))$; an index $i_u$ is returned so that $F_u[i_u].x<d(x,u)\leq F_u[i_u-1].x$. Increase counter $C$ by $F_u[i_u-1].y$. Since $V(T_u)\cap V(T_\lambda(x))$ is $\cup_{i=1}^{i_u-1}Q(u,x)$, we report sets $Q$ of all tuples in $F_u$ before position $i_u$. Note that only the preceding vertex of $x$ on $\pi_x$ has all children not on $\pi_x/{x}$, so we need to perform the queries for its each child. Once a heavy vertex is found during the traceback, we immediate terminate the queries on $\calA_1$. Clearly, $|V(T_\lambda(x))| = C$, and all $V(T_\lambda(x))$ have been correctly reported. 
    
    Since $T$ is a balanced binary tree, the counting query can be answered in $O(\log^2n)$ time. Reporting $V(T_\lambda(x))$ needs an additional time $O(K)$ where $K=|V(T_\lambda(x))|$. Recall that the preprocessing time is $O(n\log n)$. Thus, Lemma~\ref{lem:Bquerylargestsubtree} is proved. 

\subsubsection{The Data Structure $\calA_2$}\label{sec:buildingA2}

    \paragraph{\textbf{Constructing $\calA_2$}}
    The construction of $\calA_2$ for a general tree $T$ consists of the following three phases: First, transform $T$ into a binary tree $T'$ of size $O(n)$. Second, apply the \textit{spine} tree decomposition~\cite{ref:BenkocziFa03} to $T'$ to build a balance binary tree $\Gamma$ for its decomposition. Last, we construct the functions for the subtree of $T'$ in each node of $\Gamma$, similar to the above $f(v,x)$ and $Q(v,x)$, to count and determine the largest root-inclusive and leaf-inclusive subtree covered by `outside' points. The details of every phase are presented below. 
    
    \paragraph*{Phase 1.} We adapt the algorithm~\cite{ref:BanikTh16,ref:TamirAn96} to transform in $O(n)$ time $T$ into a binary tree $T'$ whose root $R(T')$ is same as $R(T)$. For completeness and correctness, the transformation is introduced below. 

    Traverse $T$ from $R(T)$ and during the traversal, 
    we process every vertex of more than two children as follows. 
    Let $v$ be such a vertex with children $v_1, v_2, \cdots, v_t$. 
    For each $v_i$ with $2\leq i\leq t-1$, we create an auxiliary vertex $u_i$ 
    of weight $w_v$ and then replace edge $e(v,v_i)$ by edge $e(u_i,v_i)$ 
    by setting $v_i$ as a child of $u_i$ with an edge of length $l(e(v,v_i))$ 
    and letting $v_i$'s parent be $u_i$. Additionally, 
    set $u_i$ as a child of the auxiliary vertex $u_{i-1}$ by an edge of length zero 
    except that $u_2$ is set as a child of $v$ by an edge of length zero. 
    Last we let $v_t$ be the other child of $u_{t-1}$ and replace 
    $v_t$'s parent by $u_{t-1}$ with $l(e(v_t,u_{t-1}))=l(e(v_t, v))$. 

    Let $V'$ denote the set of all vertices in $T'$. 
    $V'$ includes $V$ and all created auxiliary vertices but 
    $|V'| = O(n)$. 
    For any $v\in V$, we say that every auxiliary vertex created for $v$ 
    is a copy of $v$ (in the sense that they have the same weight and 
    each is connected with $v$ by a path of length zero that 
    contains only $v$'s auxiliary vertices).  
    Clearly, each $v\in V$ and all its copies $u_2, \cdots, u_{t-1}$ form a path on $T'$ so that for any $2\leq i\leq t$, $\pi(u_i,v)$ is in $\pi(u_i,R(T'))$. In addition, every edge $e$ of $T$ uniquely corresponds to an edge $e'$ on $T'$; the two incident vertices of $e'$ 
    are same as those of $e$'s, or are their copies. 

    After obtaining $T'$, we preprocess $T'$ in $O(n)$ time by computing the distance of every vertex to $R(T')$ 
    and applying the lowest common ancestor data structure~\cite{ref:BenderTh04} to $T'$ so that the distance between any two points of $T'$ can be known in constant time. 

    Consider the problem of finding the corresponding point $x'$ on $T'$ of any given point $x\in T$. Suppose $x$ is on edge $e(r,s)$. If $x$ is at $r$ or $s$, then we return that vertex containing $x$ for $x'$ due to $V\in V'$. Otherwise, we first compute the lowest common ancestor of $r$ and $s$ in constant time. If it is $r$ (resp., $s$), then $e(r,s)$ matches the edge on $T'$ that connects $s$ (resp., $r$) and its parent, which can be obtained in $O(1)$ time. Hence, $x'$ can be known in $O(1)$ time.

     Since $\calA_2$ is constructed for $T'$, it is necessary to show why one can work on $T'$ to answer the two queries about $T_\lambda(x)$ for any point $x\in T$. Denote by $T'_\lambda(x')$ the largest $x'$-inclusive subtree in $T'$ covered by the corresponding point $x'$ of $x$. We have the following observation. 
     
    \begin{observation}\label{obs:queryequivalence}
     Computing $T_\lambda(x)$ is equivalent to computing $T'_\lambda(x')$. 
    \end{observation}
    \begin{proof}
     
     We claim that $T'_\lambda(x')$ is induced by all vertices in $V(T_\lambda(x))$ and their copies in $V'$. For each vertex $v\in V(T_\lambda(x))$, $x'$ must cover in $T'$ $v$ and its all copies, i.e., their path on $T'$. Additionally, the whole subtree induced by $V(T_\lambda(x))$ and their copies is connected on $T'$. Since $x\in T_\lambda(x)$ or $x$ is interior of an edge incident to a vertex in $T_\lambda(x)$, by the above mapping between $x$ and $x'$, $x'$ is incident to a vertex of $T'_\lambda(x')$. Consequently, $V'(T'_\lambda(x'))$ consists of $V(T_\lambda(x))$ and all its copies. 

    On the other hand, every vertex in $V(T/T_\lambda(x))$ that is adjacent to 
    any vertex of $V(T_\lambda(x))$ must be a heavy vertex of $x$. 
    So, these vertices and their copies are heavy vertices of $x'$ on $T'$. 
    More importantly, on $T'$, only these vertices and their copies connect the whole subtree of $T'$ induced by $V(T_\lambda(x))$ and its copies with every vertex in $V(T/T_\lambda(x))$ and its copies on $T'$. 
    Hence, $T'_\lambda(x')$ is induced only by all vertices in $V(T_\lambda(x))$ and their copies on $T'$. Thus, the claim is true. 
    
    

    By marking every vertex of $V$ on $T'$, one can work on 
    $T'$ to find $V(T'_\lambda(x'))$ but count and report only marked 
    vertices of $V(T'_\lambda(x'))$ for querying about $T_\lambda(x)$ on $T$. Thus, the observation holds. \qed
    \end{proof}
    
    For convenience, in the following, let $T$ represent $T'$ and $V$ be $V'$. Suppose some vertices in $T$ are marked. 
    To distinguish with the original $T_\lambda(x)$, 
    we use $T'_\lambda(x)$ to represent the largest $x$-inclusive subtree 
    covered by $x$ on the current (binary) $T$ which contains marked vertices. Now, the goal is to count and report for any point $x\in T$ 
    only marked vertices of $T'_\lambda(x)$.

    \paragraph*{Phase 2.} In the second step of the preprocessing algorithm, 
    a decomposition tree $\Gamma$ for $T$ is constructed which is the base 
    of our data structure $\calA_2$. Specifically, we build $\Gamma$ by recursively 
    applying the spine decomposition~\cite{ref:BenkocziFa03} to (a rooted binary tree) $T$, 
    during which information on $T$ are collected for our needs. 
    Fig.~\ref{fig:spinedecomposition} illustrates the spine decomposition 
    of a rooted binary tree. We also present below the detailed step-by-step 
    construction of $\Gamma$. 

        \begin{figure}[h]
        \centering       
        \includegraphics[scale = 0.5]{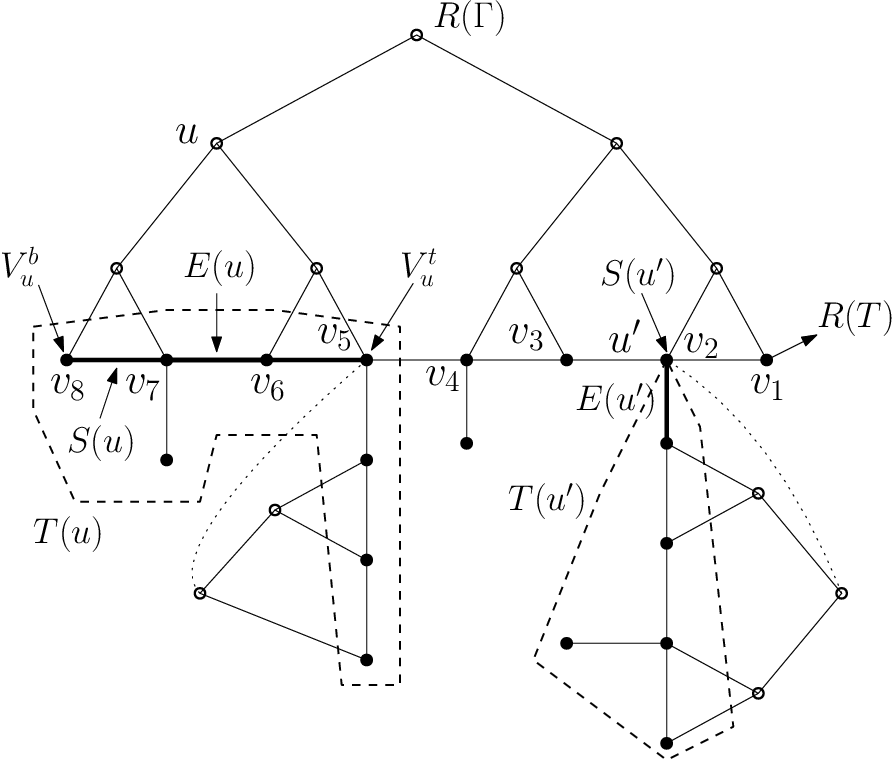}
        \caption{Illustrating the spine decomposition of a tree. On the binary search tree for the root spine from $v_1$ to $v_8$, at an internal node $u$, $S(u)$ is the subspine $\pi(v_5, v_8)$ so $V^t_u = u_5$ and $V^b_u = u_8$, $T(u)$ consists of $S(u)$ and all its hanging subtrees, $E(u) = e(v_6,v_7)$ connects the subtrees in $u$'s two children $L_u$ and $R_u$. At a leaf $u'$, $S(u')=v_2$, $T(u')$ is the left subtree rooted at $S(u')$, $E(u')$ connects $S(u')$ and its hanging subtree.}
        \label{fig:spinedecomposition}
    \end{figure}
    

    Without lost of generality, we assume that for every vertex on $T$, the subtree rooted at its right child is of size at least as same as the subtree rooted at its left child. 
    To build $\Gamma$, we traverse $T$ from its root $R(T)$. 
    By always visiting the child with most descendants, i.e., the right child, 
    a \textit{leaf} path from $R(T)$ to a leaf is thus generated, 
    and such a leaf path is called a \textit{spine} of $T$. 
    At every vertex on the generated spine, a subtree rooted 
    at its left child hangs off from the spine, 
    and it is called a \textit{hanging} subtree of this vertex. 
    The hanging subtree of a spine vertex is empty 
    if this vertex has no left child. Additionally, for any vertex $v$, we refer to the subtree rooted at $v$ containing $v$ and the subtree rooted at its left (resp., right) child as $v$'s \textit{left} (resp., \textit{right}) subtree. 
    
    Next, a balanced binary search tree is constructed on this spine so that all 
    its leaves from the left to the right represent all spine vertices 
    in the bottom-up order. At each leaf node $u$, we store the edge $E(u)$ connecting 
    the spine vertex that $u$ represents and the hanging subtree at this vertex, 
    and the subtree $T(u)$ containing $E(u)$ and this hanging subtree, 
    which is exactly the left subtree of this spine vertex. 
    Also, attach with $u$ the subspine $S(u)$ in $T(u)$, as well as the top spine vertex $V^t_u$ and the bottom spine vertex $V^b_u$ on $S(u)$, which are all that spine vertex which leaf $u$ represents. 
    For every internal node $u'$, similarly, we maintain the edge $E(u')$ 
    connecting the two distinct subtrees in its two children, 
    their union (whole subtree) $T(u')$, the subspine $S(u')$ in $T(u')$ (below $u'$), 
    and the top and bottom vertices $V^t_{u'}, V^b_{u'}$ on $S(u')$. 
    Notice that $T(u')$ is indeed the subtree generated by removing the right subtree of $V^b_{u'}$ except for $V^b_{u'}$ from the whole subtree $T_{V^t_{u'}}$ 
    rooted at $V^t_{u'}$. 
  
    Proceed with recursively decomposing every hanging subtree of this spine 
    into spines. For each obtained spine, construct a balanced binary search tree 
    on it as the above way and then set the root of its binary search tree 
    as the left child of the leaf 
    in the binary search tree of its parent spine. 
    As proved in~\cite{ref:BenkocziFa03}, 
    the obtained tree $\Gamma$ is a balanced binary 
    tree of size $O(n)$, and can be constructed in $O(n\log n)$ time. 
    
    We introduce some notations that will be used later. 
    Let $R(\Gamma)$ be the root of $\Gamma$, i.e., the root $R(\calA_2)$ of $\calA_2$. 
    For any node $u\in\Gamma$, denote by $P_u$ the whole spine containing $S(u)$, 
    let $B(u)$ represent the binary search tree constructed on spine $P(u)$ 
    whose root is denoted by $R(B(u))$. 
    Additionally, let $\rho(u)$ be the parent of $u$ on $\Gamma$, 
    and denote by $L_u$ and $R_u$ the left and right child of $u$. 

    \paragraph*{Phase 3.} In the last step of our preprocessing algorithm, 
    we construct for $T(u)$ at each node $u\in\Gamma$ 
    the following functions $f^t(u,x)$, $f^b(u,x)$, $g^t(u,x)$, $g^b(u,x)$, $Q^t(u,x)$, and $Q^t(u,x)$, all of which are defined as follows. Removing $T(u)/\{V^t_u, V^b_u\}$ from $T$ generates at most two subtrees $T_1$ and $T_2$ so that $T_1$ contains the subspine on $P(u)$ from $V^t_u$ to the top vertex of $P(u)$, and $T_2$ contains its subspine between $V^b_u$ and the leaf of $P(u)$. So, $V^t_u\in V(T_1)$ and $V^b_u\in V(T_2)$. At any $x\in [0, +\infty)$, $f^t(u,x)$ (resp., $f^b(u,x)$) 
    is the size of the largest $V^t_u$-inclusive (resp., $V^b_u$-inclusive) subtree in $T(u)$ 
    that is covered by a point in $T_1$ (resp., $T_2$) at distance $x$ to $V^t_u$ (resp., $V^b_u$); $g^t(u,x)$ (resp., $g^b(u,x)$) is the number of marked vertices on that largest $V^t_u$-inclusive (resp., $V^b_u$-inclusive) subtree in $T(u)$ at $x$.  
    
    Function $f^t(u,x)$ (resp., $f^b(u,x)$) is piecewise constant, and as $x$ increases, it breaks (decreases) only at $x$-coordinates where some vertices in $T(u)$ have their weighted distances at a point in $T_1$ (resp., $T_2$) at distance $x$ to $V^t_u$ (resp., $V^b_u$) equal to $\lambda$. Clearly, $g^t(u,x)$ (resp., $g^b(u,x)$) monotonically decreases as $x$ increases and changes only at $x$-coordinates of breakpoints of $f^t(u,x)$ (resp., $f^b(u,x)$). By this property, for any $x\in [0, +\infty)$, we define functions $Q^t(u,x)$ (resp., $Q^b(u,x)$) as the set of all additional marked vertices in the largest $V^t_u$-inclusive (resp., $V^b_u$-inclusive) subtree of $T(u)$ covered by a point in $T_1$ (resp., $T_2$) at distance $x$ to $V^t_u$ (resp., $V^b_u$) compared with the previous breakpoint of $f^t(u,x)$ (resp., $f^b(u,x)$) that has the smallest $x$-coordinate among all of $x$-coordinates larger than $x$. 
    
    We utilize two arrays $F^t_u$ and $F^b_u$ to respectively store all breakpoints of $f^t(u,x)$ and $f^b(u,x)$, their values $g^t(u,x)$ and $g^b(u,x)$, and their sets $Q^t(u,x)$ and $Q^b(u,x)$ in the descending order by theirs $x$-coordinates. Specifically, every entry in $F^t_u$ representing a breakpoint of $f^t(u,x)$ stores a tuple $(x,y,z,Q)$ where $x$ is the $x$-coordinate of this breakpoint, 
    $y$ is its value $f^t(u,x)$, $z$ is value $g^t(u,x)$, and $Q$ is its set $Q^t(u,x)$. The similar definition applies to every tuple of $F^b_u$. 
    
    Last, we compute two indices $i^t_u$ and $i^b_u$: $i^t_u$ (resp., $i^b_u$) 
    is the index of the first tuple in $F^t_u$ (resp., $F^b_u$) 
    where the whole subspine $S(u)$ in node $u$ is covered. 

    The following lemma shows the result of Phase $3$. 

    \begin{lemma}
    Phase $3$ of the preprocessing work can be done in $O(n\log n)$ time. 
    \end{lemma}
    \begin{proof}
        Traverse $\Gamma$ in the post-order and during the traversal, in the bottom-up manner, we construct 
        for each node $u$ functions $f^t(u,x)$ and $f^b(u,x)$, 
        compute values $g^t(u,x)$ and $g^b(u,x)$ as well as sets $Q^t(u,x)$ and $Q^b(u,x)$ 
        at their breakpoints, and determine indices $i^t_u$ and $i^b_u$. 
        We discuss below how to handle each case to construct these functions.    

        \begin{enumerate}
            \item $u$ is a leaf of $\Gamma$. So, $T(u)$ is a leaf or has no left child. 
            Then, $T(u) = V^t_u = V^b_u = S(u)$. 
            Clearly, at any $x\in [0,+\infty)$, $f^t(u,x) = f^b(u,x)$, $g^t(u,x) = g^b(u,x)$, 
            $Q^t(u,x) = Q^b(u,x)$, and $i^t_u = i^b_u$. 
            Note that $T(u)$ must be a marked vertex in that every auxiliary vertex 
            on $T$ is of degree $3$.  
            
            We create arrays $F^t_u$ and $F^b_u$ of size $3$, and store in order the 
            three tuples (breakpoints) $(x =+\infty, y=0, z=0, Q = \emptyset)$, 
            $(x = \frac{\lambda}{w_{T(u)}}, y=1, z=1, Q = \{T(u)\})$, 
            and $(x=0, y=1, z=1, Q = \emptyset)$ in $F^t_u$ and $F^b_u$. 
            Last, set both index $i^t_u$ and $i^b_u$ as one. Because a point at distance $F^t_u[1].x$ or $F^b_u[1].x$ covers the whole subspine $S(u)$. The running time for handling a leaf is $O(1)$ time. 
               
            \item $u$ has only one child node. So, $S(u)$ is a vertex 
            and $V^t_u = V^b_u = S(u)$. The only child of $u$ is 
            its left child $L_u$. $T(u)$ is indeed the whole subtree $T_{S(u)}$ 
            on $T$ rooted at $S(u)$ excluding the subtree rooted at $S(u)$'s 
            right child on spine $P(u)$, i.e., the left subtree of $S(u)$ 
            on $T$ which is induced by $V(T(L_u))\cup\{S(u)\}$. 
            Hence, at any $x\in [0,+\infty)$, we have $f^t(u,x)= f^b(u,x)$, 
            $g^t(u,x) = g^b(u,x)$, $Q^t(u,x)= Q^b(u,x)$, and $i^t_u = i^b_u$. 
            
            We focus on determining $f^t(u,x)$, $g^t(u,x)$, $Q^t(u,x)$, and $i^t_u$ below. 
            Since $T(L_u)$ is the whole subtree rooted at $S(u)$'s left child in $T$, clearly, $f^t(u,x)$, $g^t(u,x)$, and $Q^t(u,x)$ can be constructed in linear time by merging $F^t_{L_u}$ and such tuple array for only subtree $\{S(u)\}$ as the internal-node case for constructing $\calA_1$: If $S(u)$ is a marked vertex, then that tuple array for $\{S(u)\}$ contains only $(x = +\infty, y = 0, z = 0, Q =\emptyset), (x = \frac{\lambda}{w_{S(u)}}, y = 1, z = 1, Q ={S(u)}), (x = 0, y = 1, z = 1, Q =\emptyset)$, and otherwise, it is $(x = +\infty, y = 0, z = 0, Q =\emptyset), (x = \frac{\lambda}{w_{S(u)}}, y = 1, z = 0, Q =\emptyset), (x = 0, y = 1, z = 0, Q =\emptyset)$.  

            Due to $V^t_u = V^b_u = S(u)$, we set $i^t_u = 1$. 
            Last, copy $F^b_{u}$ to $ F^t_{u}$ and set $i^b_u = i^t_u$. 
            Hence, this case can be handled in $O(|T(u)|)$ time. 
            
            \item $u$ has two child nodes on $\Gamma$. $P(u)$, $P(L_u)$, and $P(R_u)$ 
            are the same spine. Moreover, vertices on subspine $S(L_u)$ are descendants of every vertex on $S(R_u)$; 
            $S(u)$ is the concatenation of $S(L_u)$ and $S(R_u)$ connected 
            by edge $E(u)=e(V^t_{L_u}, V^b_{R_u})$. 

            Regarding to $f^t(u,x)$, every breakpoint of $f^t(R_u,x)$ defines a breakpoint for $f^t(u,x)$. 
            In other words, $f^t(u,x)$ breaks at $F^t_{R_u}[i].x$ for each $0<i\leq |F^t_{R_u}|$. 
            Additionally, $f^t(u,x)$ breaks at $x$-coordinate $F^t_{L_u}[j].x - d(V^t_{L_u}, V^t_{R_u})$ 
            with $1\leq j\leq |F^t_{L_u}|$ if a point in subtree $\{V^t_{R_u}\}\cup T/T_{V^t_{R_u}}$ at distance 
            $F^t_{L_u}[i].x - d(V^t_{L_u}, V^t_{R_u})$ to $V^t_{R_u}$ covers 
            the whole subspine $S(R_u)$. In other words, $f^t(u,x)$ also breaks at $x = F^t_{L_u}[j].x - d(V^t_{L_u}, V^t_{R_u})$ if this value falls in $[0, F^t_{R_u}[i^t_{R_u}].x]$. 
            
            Consequently, if $i^t_{R_u}>0$, $f^t(u,x) = f^t(L_u, x + d(V^t_{L_u}, V^t_{R_u})) + f^t(R_u,x)$ 
            for $x\in[0, F^t_{R_u}[i^t_{R_u}].x]$ and $f^t(u,x) = f^t(R_u,x)$ for $x\in (F^t_{R_u}[i^t_{R_u}].x, +\infty)$; otherwise, $f^t(u,x) = f^t(R_u,x)$ for $x\in [0, +\infty)$. 
            $g^t(u,x)$ and $Q^t(u,x)$ follow the similar relation except that at $x'= F^t_{R_u}[i^t_{R_u}].x$, 
            $g^t(u,x')$ is the sum of $g^t(R_u,x')$ and values $g^t(L_u,x)$ for all tuples in $F^t_{L_u}$ whose $x$-values fall in $[F^t_{R_u}[i^t_{R_u}].x + d(V^t_{L_u}, V^t_{R_u}), +\infty)$; $Q^t(u,x') = Q^t(R_u,x')\cup\{\cup_{x=F^t_{R_u}[i^t_{R_u}].x + d(V^t_{L_u}, V^t_{R_u})}^{+\infty} Q^t(L_u,x)\}$. 
            
            By the above analysis, we can determine $f^t(u,x)$, $g^t(u,x)$, $Q^t(u,x)$, 
            and $i^t_u$ as follows. For case $i^t_{R_u}\leq 0$, we copy $F^t_{R_u}$ to $F^t_{u}$ and 
            set $i^t_u = 0$ since no point in subtree $\{V^t_{R_u}\}\cup T/T_{V^t_{R_u}}$ could cover the whole subspine $S(R_u)$. 
            For case $i^t_{R_u}>0$, we compute every breakpoint of $f^t(u,x)$ 
            and find $g^t(u,x)$ and $Q^t(u,x)$ by merging $F^t_{L_u}$ and $F^t_{R_u}$ as follows. 
            
            Use index $i$ to loop $F^t_{R_u}$ and index $j$ to loop $F^t_{L_u}$. 
            Initialize $i$ as zero but set $j$ as the index $j'$ of the entry in $F^t_{L_u}$ 
            so that $F^t_{L_u}[j'].x\leq F^t_{R_u}[i^t_{R_u}].x+d(V^t_{L_u}, V^t_{R_u}) <F^t_{L_u}[j'-1].x$, 
            which can be found in $O(\log |T(L_u)|)$ time by performing a binary search on $F^t_{L_u}$. 
            We then merge the whole array $F^t_{R_u}$ and the subarray of $F^t_{L_u}$ 
            between entry $j'$ and the last entry whose value $F^t_{L_u}[j].x -d(V^t_{L_u}, V^t_{R_u})$ is not negative. 
            The merging procedure is similar to the way for constructing $\calA_1$ 
            except that the following: Following the above relation to compute $f^t(u,x)$, $g^t(u,x)$, and $Q^t(u,x)$; in case that a breakpoint with $x = F^t_{L_u}[i^t_{L_u}].x -d(V^t_{L_u}, V^t_{R_u})$ 
            is created, that is, when $j = i^t_{L_u}$, $i\geq i^t_{R_u}$, and 
            $F^t_{L_u}[j].x - d(V^t_{L_u}, V^t_{R_u}) \geq F^t_{R_u}[i].x$, 
            we set $i^t_u$ as the index of this new breakpoint in $F^t_u$ since any point in subtree $\{V^t_{R_u}\}\cup T/T_{V^t_{R_u}}$ within distance $F^t_{L_u}[i^t_{L_u}].x -d(V^t_{L_u}, V^t_{R_u})$ to $V^t_u$ covers the whole subspine $S(u)$ at $u$. 
            Clearly, the running time is $O(|T_u|)$. 
            
            Since $f^b(u,x)$, $g^b(u,x)$, $Q^b(u,x)$, 
            and $i^b_u$ can be determined in a similar way in $O(|T_u|)$ time, we omit the details.
        \end{enumerate}

    
        In a sum, we spend $O(|T_u|)$ at each node $u\in \Gamma$ on constructing these functions. Hence, the total running time of phase $3$ is $O(n\log n)$. \qed
        \end{proof}

    Recall that the binary-tree transformation and building $\Gamma$ take $O(n\log n)$ time. 
    Combining all three phases, the following result follows. 
    
    \begin{lemma}\label{lem:constructA2}
    $\calA_2$ can be constructed in $O(n\log n)$ time. 
    \end{lemma}

    \paragraph{\textbf{Querying on $\calA_2$}} Recall that our original goal is to count and report $V(T_\lambda(x))$ for a given point $x$ on the original tree; by Observation~\ref{obs:queryequivalence}, it is sufficient to count and report only marked vertices in the largest self-inclusive subtree $T'_\lambda(x')$ covered by the corresponding point $x'$ on the binary transformation of the original tree. In the following, we shall present how to query on $\calA_2$ about the cardinality of $V(T'_\lambda(x'))$ while counting and reporting only marked vertices. 


    First, we find in $O(1)$ time the corresponding point $x'$ on the binary transformation $T$ for the given point in the original tree. Suppose $x'$ is on an edge $e(r,s)$ of $T$. Without loss of generality, assume $r$ is $s$'s parent. On $T$, at each vertex on path $\pi(s,R(T))$, a subtree rooted at the vertex hangs off $\pi(s,R(T))$ and it is referred to as the hanging subtree of the vertex, which is different with the definition for the spine decomposition. The hanging subtree of $s$ is the whole subtree $T_s$ rooted at $s$ while the hanging subtree of every other vertex on $\pi(s,R(T))$ is either its left subtree or its right subtree. Generally, for any subpath $\pi'\in\pi(s,R(T))$, a subtree containing only $\pi'$ on $\pi(s,R(T))$ hangs off $\pi(s,R(T))$; we refer to it as the hanging subtree of $\pi'$. Clearly, $T'_\lambda(x')$ contains only a subpath of $\pi(s,R(T))$ that is the longest $x'$-inclusive subpath covered by $x'$, and denote it by $\pi_{x'}$; $T'_\lambda(x')$ is exactly the largest $\pi_{x'}$-inclusive subtree covered by $x'$ in the hanging subtree of $\pi_{x'}$. 

    Denote for any vertex $v\in V$, by $u_v$ the (unique) node of $\calA_2$ with $S(u_v) = v$. Recall that $\calA_2$ is obtained by linking the binary search trees constructed for spines (leaf paths) of $T$ generated by the spine decomposition. The path $\pi(u_s,R(\calA_2))$ on $\calA_2$ passes through $t\geq 1$ binary search trees $B_1, B_2, \cdots, B_t$ in the bottom-up order. Let $\mu_i$ be the leaf of $B_i$ on $\pi(u_s,R(\calA_2))$. So, $\pi(u_s,R(\calA_2))$ is the concatenation of (disjoint) subpaths $\pi(\mu_i,R(B_i))$ for all $1\leq i\leq t$ in order. 
    
    
    For each $1\leq i\leq t$, let $\pi_i$ represent the subpath on $B_i$'s spine from $S(\mu_i)$ to its top spine vertex, that is, the subspine induced by $S(\mu_i)$ and all spine vertices in $B_i$'s leaves on the right of $\mu_i$. As illustrated in Fig.~\ref{fig:spinequery}, path $\pi(s, R(T))$ is exactly the concatenation of (disjoint) subpaths $\pi_i$ for all $1\leq i\leq t$ in order, where every two consecutive subpaths $\pi_i$ and $\pi_{i+1}$ are connected by the edge $E(\mu_{i+1})$ in $\mu_{i+1}$. 

    
    It follows that each $\pi_i$ with $1\leq i\leq t$ is the concatenation of $S(\mu_i)$ and (disjoint) subspines $S(R_u)$ for all nodes $u\in\pi(\mu_i, R(B_i))$ whose right children $R_u$ are not on $\pi(\mu_i, R(B_i))$ in the bottom-up order. Moreover, if $\mu_i=s$, then the hanging subtree of $S(\mu_i)$ off $\pi(s, R(T))$ is the union of its left and right subtrees, and otherwise, the hanging subtree of $S(\mu_i)$ is only its right subtree; for every other node $u\in\pi(\mu_i, R(B_i))$ with its right child not on $\pi(\mu_i, R(B_i))$, the hanging subtree of subpath $S(R_u)$ off $\pi(s, R(T))$ is $T(R_u)$. 



    Let $\tau$ be the largest integer with $1\leq \tau\leq t$ such that that longest $x'$-inclusive subpath $\pi_{x'}$ of $\pi(s,R(T))$ covered by $x'$ contains subpaths $\pi_1/\{s\}, \pi_2, \cdots, \pi_{\tau-1}$ but a part of $\pi_{\tau}$ or none of it. 

\begin{figure}[t]
        \centering
        \includegraphics[scale = 0.5]{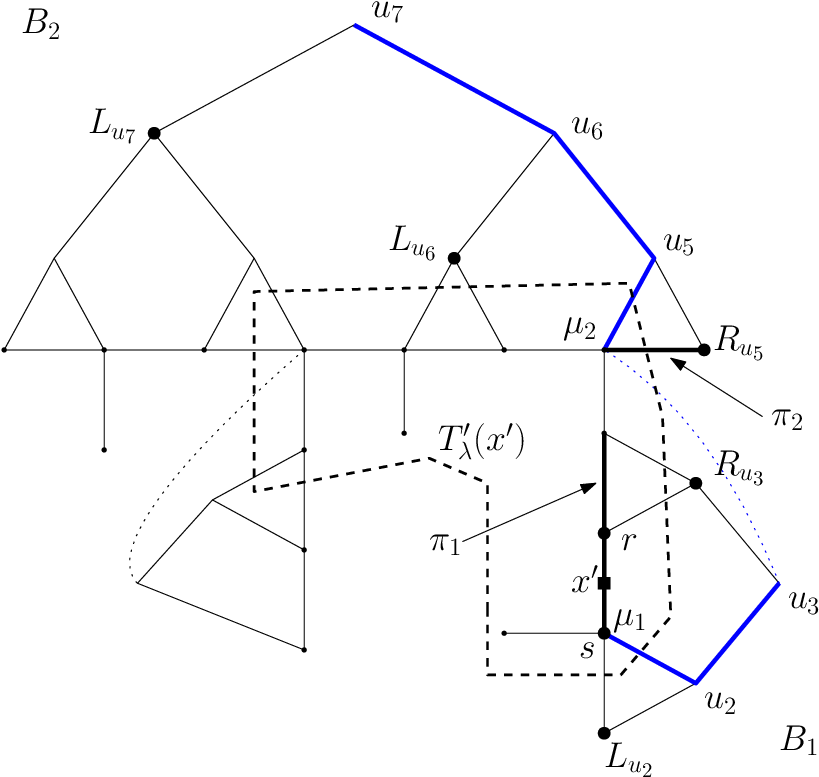}
        \caption{Illustrating the query about $T'_\lambda(x')$ on $\calA_2$ for a point $x'\in e(r,s)$: The search path for node $u_s$ is $\pi(\mu_1,u_7)$ which passes through binary search trees $B_1$ and $B_2$ on $\calA_2$ in the bottom-up order. During the traceback from $\mu_1$, we stop to visit the right child of each node above $u_5$ since $x'$ cannot cover the whose subspine in its right child $R_{u_5}$ but $T'_\lambda(x')$ contains vertices in the right subtree of $S(\mu_2)$, so the left children $L_{u_6}$ and $L_{u_7}$ are visited.}
        \label{fig:spinequery}
    \end{figure}
    
    Guided by the above analysis, $T'_\lambda(x')$ can be found as follows: For $u_s = \mu_1$, find the largest $s$-inclusive subtree covered by $x'$ in the left subtree $T(u_s)$ of $s$. For each $B_i$ encountered with $1\leq i<\tau$, we find the largest $S(\mu_i)$-inclusive subtree covered by $x'$ in $S(\mu_i)$'s right subtree, and find the largest $S(R_u)$-inclusive subtree covered by $x'$ in $T(R_u)$ for all nodes $u\in\pi(\mu_i,R(B_i))$ whose right children are not on $\pi(u_s, R(\calA_2))$. 
    For $B_\tau$, if $S(\mu_\tau)$ can be covered by $x'$, we perform the same operation as the above for $B_i$ with $i<\tau$ except that we stop querying any right child after encountering a node $u$ whose subspine $S(R_u)$ cannot be covered fully by $x'$; 
    otherwise, no vertices on $\pi_\tau$ are in $ T'_\lambda(x')$, so the procedure is terminated. 
    
    Furthermore, for any node $u\in\pi(u_s, R(\calA_2))$ with $R_u\notin\pi(u_s, R(\calA_2))$, finding the largest $S(R_u)$-inclusive subtree covered by $x'$ in $T(R_u)$ asks for finding the largest $V^b_{R_u}$-inclusive subtree containing $V^t_{R_u}$ in $T(R_u)$. This can be carried out by performing a binary search in array $F^b_{R_u}$ with key $d(x', V^b_{R_u})$. An index $i'$ is obtained so that $F^b_{R_u}[i'].y$ is the size of the largest $V^b_{R_u}$-inclusive subtree in $T(R_u)$ covered by $x'$, $F^b_{R_u}[i'].z$ is the number of all marked vertices of this subtree, and $\cup_{i=1}^{i'} F^b_{R_u}[i].Q$ is their set. 
    Additionally, if $i'\geq i^b_{R_u}$ then $x'$ covers $S(R_u)$ completely, and otherwise, it cannot. 
    
    To find the largest $S(\mu_i)$-inclusive subtree in $S(\mu_i)$'s right subtree covered by $x'$, the following lemma can be applied, which aims to find for any given vertex $v\in T$, the largest $v$-inclusive subtree in $v$'s right subtree covered by any given point $\alpha$ out of its right subtree except for $v$.

    \begin{lemma}\label{lem:rightsubtreequery}
        The counting query of the largest $v$-inclusive subtree in $v$'s right subtree covered by $\alpha$ can be answered in $O(\log h\log n)$ where $h$ is the height of the binary search tree constructed for the spine including $v$. 
    \end{lemma}
    \begin{proof}
    Recall that $B(u_v)$ is the binary search tree constructed for the spine $P(u_v)$ containing $v$, and $u_v$ is a leaf of $B(u_v)$. Consider the path $\pi(u_v, R(B(u_v)))$ from $u_v$ to $R(B(u_v))$, which is on $\pi(u_v, R(\calA_2))$. For each node $u\in\pi(u_v, R(B(u_v)))$, let $\beta(u)$ be $u$'s left child $L_u$ in $B(u_v)$ if $L_u\notin\pi(u_v, R(B(u_v)))$, and otherwise, let $\beta(u)$ be $u$.   
    

    Let $\pi'$ be the subspine of $P(u_v)$ in $v$'s right subtree. $\pi'$ is induced by $S(u_v) = v$ and all spine vertices in $B(u_v)$'s leaves on the left of $u_v$. It follows that $\pi'$ is the concatenation of (disjoint) subspines $S(u_v)$ and $S(\beta(u))$ for all nodes $u$ on $\pi(u_v, R(B(u_v)))$ with $u\neq\beta(u)$ in the bottom-up order; additionally, $T(\beta(u))$ of every node $u\in\pi(u_v, R(B(u_v)))$ with $u\neq\beta(u)$ is the hanging subtree of subspine $S(\beta(u))$ off $P(u_v)$. Hence, $v$'s right subtree is the union of $S(u_v)$ and $T(\beta(u))$ of all those nodes $u\in\pi(u_v, R(B(u_v)))$ with $u\neq\beta(u)$. 
    
    Now it is clear to see that the size of the largest $v$-inclusive subtree in $v$'s right subtree covered by $\alpha$ can be known as follows: Traverse path $\pi(u_v, R(B(u_v)))$ on $B(u_v)$ from $u_v$ to $R(B(u_v))$ in the bottom-up order. For $u_v$, we determine whether $D(\alpha,u_v)>\lambda$. If yes then the largest $v$-inclusive subtree in $v$'s right subtree covered by $\alpha$ is empty. Otherwise, we continue the traversal to visit $u_v$'s parent after incrementing the counter.
    
    For every other node $u$ with $u\neq\beta(u)$ encountered, we perform a binary search on its array $F^t_{\beta(u)}$ to find the largest $V^t_{\beta(u)}$-inclusive subtree in $T(\beta(u))$ covered by a point at distance $d(\alpha, V^t_{\beta(u)})$, and an index $i'$ is obtained. Subsequently, we increase the counter by $F^t_{\beta(u)}[i'].y$. Note that $F^t_{\beta(u)}[i'].z$ is the number of all marked vertices in that subtree, and $\cup_{i=1}^{i'} F^t_{\beta(u)}[i].Q$ is their set.

    Next, we check if $i'\geq i^t_{\beta(u)}$. If yes then $\alpha$ covers $S(\beta(u))$ entirely, and thereby, the subspine of $\pi'$ from $v$ to $V^b_{\beta(u)}$; so we continue the traversal on $\pi(u_v, R(B(u_v)))$. Otherwise, $\alpha$ cannot cover $S(\beta(u))$ entirely, so for any node $u'$ above $u$ on $\pi(u_v, R(B(u_v)))$ with $u'\neq\beta(u')$, $T(\beta(u'))$ is not relevant to that largest $v$-inclusive subtree covered by $\alpha$. Hence, the query algorithm terminates. 

    
    As a result, if $u_v\in\calA_2$ is known, the counting query of the largest $v$-inclusive subtree covered by $\alpha$ in $v$'s right subtree can be answered in $O(\log h\log n)$ time where $h$ is the height of $B(u_v)$. 

    It should be mentioned that counting and reporting only its marked vertices can be done in the same way except that for each node $u$ with $u\neq\beta(u)$, after obtaining the index $i'$, we increase the counter by $F^t_{\beta(u)}[i'].z$ rather than $F^t_{\beta(u)}[i'].y$, and additionally, report $F^t_{\beta(u)}[i].Q$ for every entry in array $F^t_{\beta(u)}$ with $i\leq i'$. The time complexities are respectively $O(\log h\log n)$ and $O(\log h\log n + K')$ where $K'$ is the number of all marked vertices. 
 
    
    It remains to find node $u_v$ on $\calA_2$. Starting from $R(\calA_2)$, we always visit the child $u$ whose associated subtree $T(u)$ contains vertex $v$ until a node of degree $2$ is encountered whose associated subspine is exactly $v$. To do so, 
    we need to address the problem that asks for any node $u$ of degree $3$ with $v\in V(T(u))$, which of the two subtrees respectively in its left and right children contains $v$. Suppose $E(u)$ is incident to vertex $s'$ and $r'$ on $T$, and $r'$ is the parent of $s'$. Because the two subtrees in $u$'s child nodes are obtained by breaking $T(u)$ at edge $E(u)$. If $s'$ is the lowest common ancestor of $v$ and $s'$, then $v$ belongs to the subtree in its left child, and otherwise, $v$ belongs to the subtree in its right child. Recall that the lowest common ancestor of any two vertices in $T$ can be known in $O(1)$ time. Hence, $u_v$ can be found in $O(\log n)$ time. 

    After locating $u_v$ in $\calA_2$ in $O(\log n)$ time, during the traceback, the above procedure can be performed to answer the counting query about the largest $v$-inclusive subtree in $v$'s right subtree covered by point $\alpha$ in $O(\log h\log n)$ time. Thus, the lemma holds. \qed
    \end{proof}

    Now we are ready to present our query algorithm on $\calA_2$ to count and report only marked vertices in $T'_\lambda(x')$. Initially, we set a flag $a$ so that $a = \text{True}$ means that the lowest binary search tree $B_\tau$ has been found, and set a flag $b$ so that $b = \text{True}$ indicates that we need to query on the left children of nodes in the binary search tree we are currently visiting. 
    
    Initialize $a$ and $b$ as False. Set counter $C =0$ and set $Q = \emptyset$ for counting and storing $V(T_\lambda(x))$, i.e., all marked vertices in $T'_\lambda(x')$. Next, find $u_s$ in $O(\log n)$ time on $\calA_2$ as the way described in Lemma~\ref{lem:rightsubtreequery}. During the traceback, starting from $u_s$, we process each node $u$ encountered on $\pi(u_s,R(\calA_2))$ as follows. 

    \begin{enumerate}
        \item $u$ is $u_s$. We compute the weighted distance $D(s,x')$ from vertex $s$ to $x'$ in constant time. If $D(s,x')\leq\lambda$, then we set $b =\text{True}$ and perform a binary search on array $F^t_u$ with key $ d(x',s)$, which returns an index $i'$. Next, we increase $C$ by $F^t_u[i'].z$, and then add $F^t_u[i].Q$ into $Q$ for all $1\leq i\leq i'$. 
        Otherwise, $s$ cannot be covered by $x'$ and so we continue to visit the parent of $u_s$ on $\pi(u_s, R(\calA_2))$. 
        
        \item $u$ is of degree $2$ but $u\neq u_s$. 
        If $a$ is true, then the last binary search tree visited on $\calA_2$ is $B_{\tau}$. Hence, we terminate the query algorithm by returning $C$ for the counting query of $V(T_\lambda(x))$ and outputting $Q$ for reporting $V(T_\lambda(x))$. 
        Otherwise, we first decide whether $D(S(u),x')>\lambda$. For yes, $B(u) = B_\tau$ but no vertices of $\pi_\tau$ are in $T'_\lambda(x')$, so we return $C$ and output $Q$ to terminate the query algorithm. 
        Otherwise, $D(S(u),x')\leq\lambda$. We set $b$ as true to give the permission to visit the left children of the following nodes on $\pi(u, R(B(u)))$ for finding the largest $S(u)$-inclusive subtree in its right subtree covered by $x'$. Additionally, if $S(u)$ is a marked vertex, we increase $C$ by one and add $S(u)$ into $Q$.
        
        \item $u$ is of degree $3$. We first determine if $u$'s right child $R_u$ is on $\pi(u_s,R(\calA_2))$, which can be known in constant time by comparing $R_u$ and the last node visited. 
        
        On the one hand, $R_u$ is on $\pi(u_s,R(\calA_2))$. So, $L_u$ is not on $\pi(u_s,R(\calA_2))$. 
        Assume that node $u'$ is the last node of degree $2$ visited. If $b$ is false, which means $u' = s$ and $D(s,x')>\lambda$, then we continue our traversal to visit the parent of $u$. 
        Otherwise, $b$ is true, so we visit its left child $L_u$ for finding the largest $S(u')$-inclusive subtree covered by $x'$ in $S(u')$'s right subtree. As in Lemma~\ref{lem:rightsubtreequery}, we perform a binary search on $F^t_{L_u}$ with key $d(x', V^t_{L_u})$, which returns an index $i'$; then increase $C$ by $F^t_{L_u}[i'].z$ and add $F^t_{L_u}[i].Q$ for all $1\leq i\leq i'$ to $Q$. Next, we compare indices $i'$ and $i^t_{L_u}$. For case $i'<i^t_{L_u}$, we set $b$ as false since $T'_\lambda(x')$ does not contain the entire subspine $S(L_u)$. 

        On the other hand, $R_u$ is not on $\pi(u_s,R(\calA_2))$. If $a$ is true, then $b$ must be true. Because $B(u)$ is $B_\tau$ and we have found the node $u''$ on $\pi(\mu_\tau, R(B_\tau))$ whose subspine $S(R_{u''})$ cannot be fully covered by $x'$. 
        We thus continue our traversal on $\pi(u_s,R(\calA_2))$ in order to find the largest $S(\mu_\tau)$-inclusive subtree covered by $x'$ in $S(\mu_\tau)$'s right subtree. 
        Otherwise, $a$ is false, so $B_\tau$ is above $B(u)$ on $\calA_2$. We perform a binary search on $F^b_{R_u}$ with key $d(x', V^b_{R_u})$. By the obtained index $i'$, we increase $C$ by $F^b_{R_u}[i'].z$ and insert $F^b_{R_u}[i].Q$ to $Q$ for every $i\leq i'$. Further, we check if $i'< i^b_{R_u}$. If yes, $x'$ cannot cover the entire subpath $S(R_u)$ on $\pi(s,R(T))$, which means $B(u)$ is $B_\tau$; hence, we set $a$ as true. 
    \end{enumerate}

    In worst case, our query algorithm spends $O(\log n)$ time on processing each node on $\pi(u_s, R(\calA_2))$ for counting all marked vertices in $T'_\lambda(x')$, i.e., counting $V(T_\lambda(x))$. Hence, the counting query of $V(T_\lambda(x))$ can be answered in $O(\log^2 n)$ time. While counting $V(T_\lambda(x))$, only marked vertices of $T'_\lambda(x')$ are inserted into $Q$. In total, additional $O(|V(T_\lambda(x))|)$ time is spent on reporting $V(T_\lambda(x))$. Thus, we have the following result. 

    \begin{lemma}\label{lem:countA2}
    The counting and reporting queries of $V(T_\lambda(x))$ can be answered respectively in $O(\log^2n)$ and $O(\log^2n +K)$ time where $K=|V(T_\lambda(x))|$.  
    \end{lemma}




    In a sum, with $O(n\log n)$-time preprocessing work on constructing $\calA_2$, the counting query about $V(T_\lambda(x))$ can be answered in $O(\log^2n)$ time and all its vertices can be found in $O(\log^2n +K)$ time. Hence, Lemma~\ref{lem:querylargestsubtree} is proved. 


\subsubsection{The Proof of Corollary~\ref{cor:fasterkquery}}\label{sec:improveA1A2}
    Recall that to determine the feasibility of any given $\lambda$, it is sufficient to decide whether there exists a point in $T$ such that the largest self-inclusive subtree covered by it is of size no less than $k$. To this end, for any given point $x$ to test, we only need to decide if $T_\lambda(x)$ has at least $k$ vertices, which leads the following improvements on $\calA_1$ and $\calA_2$. 
    

    When $T$ is a balanced binary search tree, recall that for each vertex $v$ of $T$, we construct function $f(v,x)$ that determines the size of the largest $v$-inclusive subtree covered by $x$ in the subtree $T_v\in T$ rooted at $v$, where $x$ is an `outer' point in $\{v\}\cup T/T_v$ at distance $x$ to $v$. Regarding to the above goal, instead of computing all breakpoints of $f(v,x)$, we only need to find its $k$ smallest breakpoints larger than zero (that is, the breakpoints of the $k$ largest $x$-coordinates smaller than $+\infty$ since $f(v,x)$ monotonically increases as $x$ decreases). The time complexity of constructing $\calA_1$ is still $O(n\log n)$. While determining the feasibility of the given $\lambda$, for any given point $x$ to test, we perform the same algorithm for counting $T_\lambda(x)$ on $\calA_1$ but terminate the query algorithm once the counter becomes no less than $k$. Clearly, we can decide whether $|T_\lambda(x)|\geq k$ in $O(\log n\log k)$ time. As a consequence, for $T$ being a balanced binary search tree, the time complexity of the feasibility test is improved to $O(n\log n\log k)$. 

    When $T$ is a general tree, recall that for each node $u$ of $\calA_2$, array $F^t_u$ (resp., $F^b_u$) is computed that stores all breakpoints of function $f^t(u,x)$ (resp., $f^b(u,x)$), which computes the size of the largest $V^t_u$-inclusive (resp., $V^b_u$-inclusive) subtree covered by point $x$ in $u$'s associated subtree $T(u)\in T$, and values $g^t(u,x)$ (resp., $g^b(u,x)$) at (the $x$-coordinate of) each of these breakpoints, which computes the number of marked vertices in that largest $V^t_u$-inclusive (resp., $V^b_u$-inclusive) subtree covered by $x$, as well as set $Q^t(u,x)$ (resp., $Q^b(u,x)$) of additional marked vertices covered by $x$, where $x$ is an `outer' point in subtree $V^t_u\cup T/T_{V^t_u}$ (resp., subtree $T_{V^b_u}$) at distance $x$ to $V^t_u$ (resp., $V^b_u$). 
    
    Regarding to the feasibility test, in the preprocessing, for each node $u$ of $\calA_2$, we construct only functions $f^t(u,x)$ and $g^t(u,x)$ as well as functions $f^b(u,x)$ and $g^b(u,x)$ for $T(u)$ as in Lemma~\ref{lem:constructA2}. Because $g^t(u,x)$ (resp., $g^b(u,x)$) is a piecewise constant function that monotonically increases as $x$ decreases, and breaks only at a subset of $x$-values of breakpoints of $f^t(u,x)$ (resp., $f^b(u,x)$). 
    For each $u\in\calA_2$, we scan array $F^t_u$ (resp., $F^b_u$) to compute a sequence that contains indices of all entries in $F^t_u$ (resp., $F^b_u$) where $g^t(u,x)$ (resp., $g^b(u,x)$) breaks at their $x$-values and is larger than zero, but store only the first $k$ indices in another array $G^t_u$ (resp., $G^b_u$). Indeed, $G^t_u$ (resp., $G^b_u$) represents the ordered set of the $k$ smallest breakpoints of $g^t(u,x)$ (resp., $g^b(u,x)$) larger than zero. Clearly, the preprocessing time remains $O(n\log n)$. 
    

    Furthermore, while determining the feasibility of $\lambda$, for any given point $x$, we compute its corresponding point $x'$ on the binary transformation and then perform the algorithm for counting $T'_\lambda(x')$ as in Lemma~\ref{lem:countA2} except that for each node $u\in \calA_2$ to process, we perform the binary search on array $G^t_{u}$ or $G^b_{u}$, instead of array $F^t_{u}$ or $F^b_{u}$, 
    to decide how many marked vertices of $T(u)$ are in $T'_\lambda(x')$. Specifically, for each entry in $G^t_{u}$ or $G^b_{u}$ to access, we visit its corresponding entry in $F^t_{u}$ or $F^b_{u}$ and compare the $x$-value in the stored tuple with $d(x', V^t_u)$ or $d(x', V^b_u)$, accordingly.

    
    
    An index $i'$ with $i'\leq k$ is obtained at the end. By its corresponding entry in $F^t_{u}$ or $F^b_{u}$, 
    we increase the counter by the $z$-value of the stored tuple, and then check if the counter is no less than $k$. If yes, $\lambda$ is feasible and hence we terminate the query algorithm. 
    Otherwise, we find the entry $j'$ in $F^t_{u}$ or $F^b_{u}$ corresponding to the next entry $i'+1$ of $G^t_{u}$ 
    or $G^b_{u}$, which can be done in $O(1)$ time. Note that $i'<k$ since otherwise, the counter is at least $k$. Next, we compare $j'-1$ with $i^t_{u}$ or $i^b_{u}$ accordingly 
    to determine if $T'_\lambda(x')$ contains the whole subspine $S(u)$. This is because the total number of (marked and unmarked) vertices 
    in $T(u)$ that belong to $T'_\lambda(x')$ is $F^t_{u}[j'-1].y$ or $F^b_{u}[j'-1].y$, accordingly. 
    
    It is clear to see that during the traceback on path $\pi(u_s,R(\calA_2))$, we spend only $O(\log k)$ time on 
    processing each node. Hence, for any given $x$ of $T$, we can decide whether $|V(T_\lambda(x))|\geq k$ 
    in $O(\log n\log k)$ time. It follows that the feasibility of any given $\lambda$ 
    can be decided in $O(n\log n\log k)$ time. Thus, the corollary is proved. 



\subsection{Computing $\lambda^*$}\label{sec:treecomputelambda}

Observation~\ref{obs:lambda_belongs_oneset} in Section~\ref{sec:preliminary} shows that 
$\lambda^*$ is in the set $\Lambda$ that consists of values generated by solving $D(v,x) = D(u,x)$ for every pair of vertices $u,v$ w.r.t. $x\in \pi(u,v)$. 
Unlike the graph version, we employ the centroid decomposition~\cite{ref:MegiddoAn81} to 
implicitly enumerate every pair of vertices such that only $O(n\log n)$ linear functions $y = D(v,x)$ are needed so that $\Lambda$ belongs to the set of the $y$-coordinates of all their intersections. 

The \textit{centroid} of a tree is a vertex at which the tree can be decomposed into three or fewer subtrees with only this common vertex such that each of them is of size at most half of the tree. The centroid of a tree and these subtrees can be found in linear time~\cite{ref:GoldmanMi72,ref:MegiddoAn86}. 
As shown in~\cite{ref:MegiddoAn81}, by recursively decomposing $T$ at its centroid, a decomposition tree of height $O(\log n)$ can be constructed in $O(n\log n)$ time: $T$ is stored in the root; for each internal node, the subtrees of $T$ that are stored in all its (at most three) children nodes are generated by decomposing the subtree in this node at the centroid; every leaf maintains a vertex of $T$ uniquely. 

Similarly, we find the centroid $c$ of $T$ and decompose it into three subtrees $T_1$, $T_2$, and $T_3$. 
Consider $c$ at the origin of the $x,y$-coordinate plane. 
For each vertex $v\in V$, let $v_l$ and $v_r$ be the two points
on the $x$-axis at distance $d(c,v)$ to $c$ (the origin) respectively on its left and right, 
and we construct the distance function $y = w_v\cdot(x-x(v_l))$ for $v_l$ 
and $y = w_v\cdot(x(v_r)-x)$ for $v_r$. The set of the $y$-coordinates 
of all intersections between the obtained $2n$ lines contains all values in $\Lambda$ 
caused by every pair of vertices from different subtrees of $T_1$, $T_2$, and $T_3$.

We recursively decompose $T_1$, $T_2$, and $T_3$ respectively at their centroids, 
and construct distance functions for vertices w.r.t. each centroid. Thus, in $O(n\log n)$ time, $O(n\log n)$ linear functions are derived so that $\Lambda$ belongs to the set of the $y$-coordinates of their intersections. 

At this point, we can adapt Lemma~\ref{lem:linearrangementtechnique} 
with the assistance of Lemma~\ref{lem:treedecision} to find $\lambda^*$ 
among the $y$-coordinates of all intersections between the $O(n\log n)$ lines, which runs totally in $O(n\log^2 n\log k)$ time. 
The following result is thus obtained. 

\begin{theorem}\label{the:weightedtree}
    The weighted connected $k$-vertex one-center problem on trees can be solved in $O(n\log^2 n\log k)$ time. 
\end{theorem}

When the vertices of $T$ are all weights one, the following theorem computes $\lambda^*$ and $x^*$ in $O(n\log^2 n)$ time. 
   
\begin{theorem}\label{the:unweightedtree}
    The unweighted connected $k$-vertex one-center problem on trees can be solved in $O(n\log^2 n)$ time.  
\end{theorem}
\begin{proof}
Observation~\ref{obs:lambda_belongs_oneset} and~\ref{obs:smallestdiameterequivalency} imply that the unweighted tree version aims to compute a $k$-subtree of smallest diameter so that its diameter, which is its longest path length, is exactly $2\lambda^*$, and $x^*$ is the center of this longest path. Hence, computing $\lambda^*$ is equivalent to solving the problem of finding the vertex so that the distance from it to its $k$-th closest vertex is smallest.

Our algorithm is simple and consists of three steps: 
First, we implicitly form the intervertex distance subsets for each vertex by employing the centroid decomposition~\cite{ref:MegiddoAn81}. Second, we find the length of that $k$-th shortest path for every vertex. Last, we compute the smallest value of them, which is exactly $2\lambda^*$, and find the center $x^*$. 

The intervertex distance subsets for each vertex of $V$ are 
implicitly formed in a similar way as in~\cite{ref:MegiddoAn81}, 
which is for computing the $k$-th longest path on a tree. 
More specifically, $T$ is decomposed at its centroid $c$ into three or 
fewer subtrees, e.g., $T_1$, $T_2$, and $T_3$. Then, three sorted subsets 
$L_1, L_2$ and $L_3$ are explicitly formed so that 
for each $1\leq i\leq 3$, $L_i$ is the ordered set of the distances from $c$ to every other vertex of $T_i$. 
Next, for each vertex $v$ of $T_1/\{c\}$, a sorted subset is implicitly formed 
for $v$'s distance to every vertex in $T_2/\{c\}$ by adding $d(v,c)$ to $L_2$, 
and another sorted subset is implicitly formed for $v$'s distances 
to every vertex of $T_3/\{c\}$ by adding $d(v,c)$ to $L_3$. 
Additionally, for each vertex of $T_2/\{c\}$, two sorted subsets are implicitly 
generated for its distances to vertices respectively in $T_1/\{c\}$ 
and $T_3/\{c\}$. Similarly, two sorted subsets are implicitly formed 
for each vertex in $T_3/\{c\}$. Clearly, at most $3n$ sorted subsets are generated but the total storage space is $O(n)$. Due to the sorting work on $L_1$, $L_2$, and $L_3$, the time complexity is $O(n\log n)$. 

We proceed to recursively decompose each of the three subtrees $T_1$, $T_2$, and $T_3$ at their own centroids, and form sorted subsets as the above for each subtree. To the end, $O(n\log n)$ sorted subsets are generated, which implicitly enumerate the intervertex distances for every vertex of $V$ but take $O(n\log n)$ space in total. Note that the value of every entry in any subset can be known in constant time. 

Further, we compute for every vertex the length of its $k$-th shortest path. More specifically, for each vertex $v$, let $n_v$ represent the number of all intervertex distance subsets of $v$. Since every subset of $v$ is sorted and each entry can be accessed in constant time, 
its $k$-th shortest path length can be found in $O(n_v\log n)$ time by the algorithm~\cite{ref:FredericksonTh82}, 
which is for finding the $k$-th smallest value of multiple sorted arrays. Among all obtained values, we set $\lambda^*$ as the smallest one. The total running time is $O(\sum_{v\in V} n_v\log n)$, which is $O(n\log^2 n)$. 

By the corresponding entry of $\lambda^*$ in the subsets, the two vertices and their path that decide $\lambda^*$ can be obtained in $O(n)$ time. Consequently, $x^*$, which is exactly the center of this path, can be found in $O(n)$ time. It follows that the optimal $k$-subtree $T^*$, induced by the $k$ closest vertices of $x^*$, can be obtained in $O(n)$ time by reporting $k$ vertices within distance $\lambda^*$ to $x^*$ during the pre-order traversal on $T$ from $x^*$. 

In a sum, we can find the partial center $x^*$ of $T$ and $T^*$ in $O(n\log^2 n)$ time. The theorem thus follows.\qed
\end{proof}




\section{Conclusion}

\begin{figure}
    \centering
\includegraphics[width=0.26\textwidth]{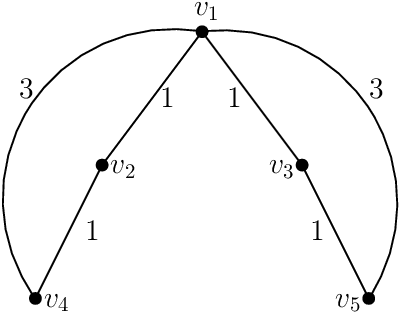}
        \caption{Illustrating Corollary~\ref{cor:proveTsmallestdiameter} does not apply to the general weighted case by this example: $k=3$, $w_{v_1} =100$, $w_{v_2}=w_{v_3} =50$, 
        $w_{v_4}=w_{v_5} =1$, and all edge lengths are as labeled.  Clearly, $x^* = v_1$, $\lambda^*=50$, and $T^*$ is induced by $\{v_1, v_2, v_3\}$, so $W(T^*) = 50$; however, the subgraph induced by $\{v_1, v_2, v_3\}$ is a $3$-subtree of smallest (weighted) diameter which is $3$.}
        \label{fig:counterexample}
\end{figure}

Recall that in Section~\ref{sec:preliminary}, we proved in Corollary~\ref{cor:proveTsmallestdiameter} that for the unweighted version or $G$ being a tree, solving the problem is equivalent to finding the $k$-subtree of smallest (weighted) diameter in $G$. However, it does not apply to the weighted version on a general graph. See Fig.~\ref{fig:counterexample} for an example where the optimal $k$-subtree $T^*$ covered by $x^*$ is not of smallest weighted diameter among all $k$-subtrees. To the best of our knowledge, no previous work has been found for finding a $k$-subtree in a general graph with smallest weighted diameter. So, it is quite interesting to explore this graph problem in future.

\bibliographystyle{splncs04}

\end{document}